\documentclass[11pt]{amsart}
\usepackage{fullpage,graphicx,mathpazo,color}
\usepackage{amsmath,amscd,amssymb}
\usepackage{mathrsfs}
\usepackage{subfigure}
\usepackage{epstopdf}
\usepackage{amsmath}
\usepackage{tikz}
\allowdisplaybreaks[4]
\usepackage[colorlinks, citecolor=blue,pagebackref,hypertexnames=false]{hyperref}

\begin{document}

\newtheorem{theorem}{Theorem}[section]
\newtheorem{lemma}{Lemma}[section]
\newtheorem{proposition}{Proposition}[section]
\newtheorem{corollary}{Corollary}[section]
\newtheorem{remark}{Remark}[section]
\numberwithin{equation}{section}

\newcommand{\bfR}{{\Bbb R}}
\newcommand{\bfC}{{\Bbb C}}
\newcommand{\bfZ}{{\Bbb Z}}
\newcommand{\ii}{\text{i}}
\newcommand{\e}{\text{e}}
\newcommand{\dd}{\text{d}}
\newcommand{\Om}{\omega}
\newcommand{\nn}{\nonumber}
\newcommand\be{\begin{equation}}
\newcommand\ee{\end{equation}}
\newcommand{\bea}{\begin{eqnarray}}
\newcommand{\eea}{\end{eqnarray}}
\newcommand\berr{\begin{eqnarray*}}
\newcommand\eerr{\end{eqnarray*}}

\newcommand{\change}[1]{{\color{blue}#1}}
\newcommand{\note}[1]{{\color{red}#1}}

\title{A New Two-component Sasa--Satsuma Equation: Large-time Asymptotics on the Line}

\author{Xiaodan Zhao}
\address{ School of Control and Computer Engineering, North China Electric Power University, Beijing 102206, P.R. China}
 \email{xiaodan\_zhao0224@163.com (X. Zhao)}

\author{Lei Wang$^*$}
\address{School of Mathematics and Physics, North China Electric Power University, Beijing 102206, P.R. China}
\email{50901924@ncepu.edu.cn (L. Wang)}


\begin{abstract}
We consider the initial value problem for a new two-component Sasa--Satsuma equation associated with $4\times4$ Lax pair with decaying initial data on the line. By utilizing the spectral analysis, the solution of the new two-component Sasa--Satsuma system is transformed into the solution of a $4\times4$ matrix Riemann--Hilbert problem. Then the long-time asymptotics of the solution is obtained by means of the nonlinear steepest descent method of Deift and Zhou for oscillatory Riemann--Hilbert problems. We show that there are three main regions in the half-plane $-\infty<x<\infty$, $t>0$, where the asymptotics has qualitatively different forms: a left fast decaying sector, a central Painlev\'e sector where the asymptotics is described in terms of the solution of a new coupled Painlev\'e II equation which is related to a $4\times4$ matrix Riemann--Hilbert problem, and a right slowly decaying oscillatory sector.

\end{abstract}
\date{}

\thanks{$^*$Corresponding author.}
\maketitle

\section{Introduction}
The one-dimensional cubic NLS (nonlinear Schr\"odinger) equation
\be
\ii q_T+\frac{1}{2}q_{XX}+|q|^2q=0,
\ee
is a universal model for the evolution of quasi-monochromatic waves in weakly nonlinear dispersive media \cite{BN-JMP}. It has important applications in many different physical contexts, such as deep water waves, nonlinear fiber optics, acoustics, plasma physics, Bose--Einstein condensation (e.g., see \cite{APT,CGT,P1983,S-Sulem} and references therein). One of the most successful related applications is the description of optical solitons in fibres. However, in order to illustrate the propagation of the ultrashort (femtosecond) optical pulses, the NLS equation becomes less accurate \cite{R-OL}, and thus some additional effects such as the third-order dispersion, self-steepening and stimulated Raman scattering should be added to meet this requirement. In this setting, Kodama and Hasegawa \cite{KH} proposed the higher-order NLS equation
\be\label{HONLS}
\ii q_T+\frac{1}{2}q_{XX}+|q|^2q
+\ii\varepsilon\left\{\beta_1q_{XXX}+\beta_2|q|^2q_X+\beta_3q(|q|^2)_X\right\}=0,
\ee
where $\varepsilon$ is a small parameter and represents the integrable perturbation of the NLS equation, $\beta_1$, $\beta_2$ and $\beta_3$ are real parameters.

In general, Equation \eqref{HONLS} is not completely integrable, unless certain restrictions are imposed
on $\beta_1$, $\beta_2$ and $\beta_3$. In particular, when the ratio of their coefficients satisfies $\beta_1:\beta_2:\beta_3=1:6:3$, the Equation \eqref{HONLS} can be immediately reduced to the well-known integrable Sasa--Satsuma equation \cite{SS} in the form
\be\label{SS-Eq}
\ii q_T+\frac{1}{2}q_{XX}+|q|^2q
+\ii\varepsilon\left\{q_{XXX}+6|q|^2q_X+3q(|q|^2)_X\right\}=0.
\ee
For the convenience of analyzing the Sasa--Satsuma equation \eqref{SS-Eq}, according to \cite{SS}, one can introduce variable transformations
\begin{align}
u(x,t)=q(X,T)\exp\left\{-\frac{\ii}{6\varepsilon}\left(X-\frac{T}{18\varepsilon}\right)\right\},
\ t=T,\ x=X-\frac{T}{12\varepsilon},
\end{align}
then Equation \eqref{SS-Eq} reduces to a complex modified KdV (Korteweg--de Vries)-type equation
\begin{equation}\label{1.1}
u_t+\varepsilon\{u_{xxx}+6|u|^2u_x+3u(|u|^2)_x\}=0.
\end{equation}

On account of its integrability and physical implications, the Sasa--Satsuma equation has attracted much attention and various works have been presented since it was discovered. For instance, the double hump soliton solutions of the Sasa--Satsuma equation have been obtained in \cite{MTMPT,SS} by means of inverse scattering approach. While, its multi-soliton solutions have been constructed in \cite{GHNO} by the Kadomtsev--Petviashvili hierarchy reduction method. Besides, the Darboux transformation \cite{LLM} and the RH (Riemann--Hilbert) problem approach \cite{YC} were also imposed separately on this equation to obtain the high-order soliton solutions. Moreover, breather and rogue wave solutions for Sasa--Satsuma equation were also derived \cite{ASDH,CSH,MQ,WWSF}. In addition to the initial-boundary value problem for Sasa--Satsuma equation on the half-line and a finite interval were also investigated via the unified transform method in \cite{JX-PA} and \cite{XZF}, respectively. Beyond that, the long-time asymptotic behaviour of the solution to Sasa--Satsuma equation \eqref{1.1} with decaying initial data were analyzed in \cite{LGX,LG-SS} and \cite{HL-JDE} respectively in the sectors $0<c_1<x/t<c_2$ and $|x|\leq c_3t^{1/3}$ by using the nonlinear steepest descent method for oscillatory Riemann--Hilbert problems. Very recently, data-driven solutions and parameter discovery of the Sasa--Satsuma equation was studied via the physics-informed neural networks method in \cite{LWZ-PhysD}.

Since various complex systems such as multimode or wavelength-division multiplexing fibers usually involve more than one component, the studies should be extended to multi-component Sasa--Satsuma equation cases \cite{G-MMAS,LTYD,NPSM,W-ND,ZWM}. Based on this fact, in present paper, we will consider a new integrable two-component Sasa--Satsuma equation \cite{G-MMAS,H-AMP,W-ND},
\be\label{SS}
\begin{aligned}
u_t&=u_{xxx}+6|u|^2u_x+3(|u|^2)_xu+3w(uw)_x,\\
w_t&=w_{xxx}+6|u|^2w_x+3(|u|^2)_xw+6w^2w_x,
\end{aligned}
\ee
where $u$ is a complex-valued function, $w$ is a real-valued function. It is readily to see that when $w=0$, the new two-component Sasa--Satsuma equation \eqref{SS} can be reduced to the Sasa--Satsuma equation \eqref{1.1} with $\varepsilon=-1$. On the basis of spectral analysis of the $4\times4$ matrix Lax pair for the two-component Sasa--Satsuma equation, the $N$-soliton formulas expressed by the ratios of determinants were discussed in \cite{W-ND} via the RH approach, moreover, the traveling soliton, breather soliton and rogue wave solutions have been constructed by Darboux transformation method in \cite{G-MMAS}. Recently, the initial-boundary value problem of Equation \eqref{SS} on the half-line has been solved with the aid of the unified transformation method \cite{H-AMP}.

The inverse scattering transform based on the RH problem is a very powerful tool in the study of the nonlinear integrable equations. It can obtain explicit soliton solutions for the integrable systems under reflectionless potentials condition. However, as is well-known, one can not solve the RH problems in a closed form unless in the case of reflectionless potentials. As a consequence, the study on large-time asymptotic behavior of solutions becomes an attractive topic in integrable systems. There were a number of progresses in this formidable subject \cite{AS-SAM,Its,ZM}, nevertheless, the nonlinear steepest descent method for oscillatory Riemann--Hilbert problems proposed by Deift and Zhou \cite{PD} turned out to be a great achievement in the further development of analyzing the long-time asymptotics for the initial value problems of integrable nonlinear evolution equations. Up to now, with this method, numerous new significant long-time asymptotic results for various nonlinear completely integrable models associated with $2\times2$ matrix spectral problems were obtained in a rigorous and transparent form (see \cite{AL-Non,BJM,BIK-CMP,RB,CL,LTYF,LG-JDE,WGC-PhysD,XJ-JDE}). Recently, the long-time asymptotics
for some integrable nonlinear evolution equations associated with higher-order matrix Lax pairs were studied in accordance with the procedures of the Deift--Zhou nonlinear steepest descent method, such as
Degasperis--Procesi equation \cite{BS-DP}, coupled NLS equation \cite{GL-JNS}, Sasa--Satsuma equation \cite{HL-JDE,LGX,LG-SS}, Spin-1 Gross--Pitaevskii equation \cite{GWC-CMP}, matrix modified KdV equation \cite{LZG}, three-component coupled NLS system \cite{MWX} and so on.

The main goal of the present paper is to extend the nonlinear steepest descent method to study the long-time asymptotic behavior for the Cauchy problem of the new two-component Sasa--Satsuma equation \eqref{SS} associated with a $4\times4$ Lax pair on the line with the initial data
\be\label{1.7}
u(x,0)=u_0(x),\quad w(x,0)=w_0(x),
\ee
where $u_0(x)$ and $w_0(x)$ belong to the Schwartz space $\mathcal{S}(\bfR)$. Our first step is to formulate the main matrix RH problem corresponding to Cauchy problem \eqref{SS}-\eqref{1.7}. The most outstanding structure of this system is that it admits a $4\times4$ matrix spectral problem, however, all the $4\times4$ matrices in this paper can be rewritten as $2\times2$ block ones. Thus we can directly formulate the $4\times4$ matrix RH problem by the combinations of the entries in matrix-valued eigenfunctions instead of using the Fredholm integral equation to construct another set of eigenfunctions \cite{BS-DP,CJ-PD}. As a consequence, a RH representation of the solution of the Cauchy problem \eqref{SS}-\eqref{1.7} is given (Theorem \ref{th2.1}). Then, this representation obtained allows us to apply the nonlinear steepest descent method for the associated $4\times4$ matrix RH problem and to obtain a detailed description for the leading-order term of the asymptotics of the solution.

We will first consider the asymptotic behavior of the solution in oscillatory sector characterized by \eqref{3.1} (Theorem \ref{th3.2}). It is noted that the phase function $\Phi(x,t;\xi)$ of $\e^{\pm t\Phi(x,t;\xi)}$ involved in the jump matrix has two stationary points in this region. This immediately leads us to introduce a $3\times3$ matrix-valued function $\delta(\xi)$ function to remove the middle matrix term when we split the jump matrix into an appropriate upper/lower triangular form. However, the function $\delta$ cannot be solved explicitly since it satisfies a $3\times3$ matrix RH problem. Recalling that the topic of our paper is studying the asymptotic behavior of solution, we can replace function $\delta(\xi)$ with $\det\delta(\xi)$ by adding an error term by following the idea first introduced in \cite{GL-JNS}. Then the exact solution of a class of model RH problem that is relevant near the critical points is derived (Theorem \ref{th3.1}), which generalizes the model problem considered in \cite{GL-JNS,LGX,LG-SS,MWX} and also can be used to analyze the long-time asymptotics of other integrable models. Next, we study the asymptotics of solution in Painlev\'e sector given in \eqref{4.1} (Theorem \ref{th4.2}). Our main result shows that the leading-order asymptotics for Equation \eqref{SS} is depicted in terms of the solution of a new coupled Painlev\'e II equation \eqref{B.5}. Interestingly, we noticed that the functions $u_p(y)$ and $w_p(y)$ of solution of \eqref{B.5} are complex-valued and real-valued functions, respectively, moreover, $u_p(y)$ has constant phase, which is very different from the result obtained for the matrix modified KdV equation in same region \cite{LZG}. This innovative work enriches the Painlev\'e asymptotic theory in the field of long-time dynamic analysis for integrable systems. Finally, the asymptotic behavior of the solution in the fast decay sector \eqref{5.1} is derived by performing a trivial contour deformation (Theorem \ref{th5.1}).

The organization of this paper is as follows. In Section \ref{sec2}, a basic $4\times4$ matrix RH problem with the aid of the inverse scattering method is constructed, whose solution gives the solution of the initial value problem \eqref{SS}-\eqref{1.7}, where the Lax pair of the new two-component Sasa--Satsuma equation and the relevant matrices are written as block forms. Sections \ref{sec3}-\ref{sec5} perform the asymptotic analysis of this RH problem leading to asymptotic formulas for the solution. We mainly analyze three regions in the $(x,t)$-half-plane where the asymptotic behavior of the solution is qualitatively different: (i) A slowly decaying oscillatory sector (Theorem \ref{th3.2}), (ii) A Painlev\'e sector (Theorem \ref{th4.2}), (iii) A fast decay sector (Theorem \ref{th5.1}). Appendix \ref{secA} is devoted to give the proof of Theorem \ref{th3.1}. The RH problem associated with the new coupled Painlev\'e II equation is discussed in Appendix \ref{secB}.

\section{Basic Riemann--Hilbert problem}\label{sec2}

An essential ingredient in the following analysis is the $4\times4$ matrix Lax pair \cite{W-ND} of the new two-component Sasa--Satsuma equation \eqref{SS}, which reads
\begin{align}
&\psi_x=U\psi=(-\ii\xi\sigma+U_1)\psi,\label{2.1}\\
&\psi_t=V\psi=(4\ii\xi^3\sigma+V_1)\psi,\label{2.2}
\end{align}
where $\psi(x,t;\xi)$ is a $4\times4$ matrix-valued function, $\xi$ is the spectral parameter,  $\sigma=\text{diag}(-1,1,1,1)$. The $4\times4$ matrix-valued functions
\begin{align}
U_1(x,t)=&\begin{pmatrix}
0 & q \\
-q^\dag & \mathbf{0}_{3\times3}
\end{pmatrix},\quad q(x,t)=\begin{pmatrix}-u&-u^*&-w\end{pmatrix},\label{2.3}\\
V_1(x,t;\xi)=&-4\xi^2U_1+2\ii\xi(U_1^2+U_{1x})\sigma+[U_1,U_{1x}]+U_{1xx}-2U_1^3.\label{2.4}
\end{align}
where $``*"$ and $``\dag"$ denote complex conjugation of a complex number and Hermitian conjugation of a complex matrix or vector, respectively. A direct calculation shows that the zero-curvature equation $U_t-V_x+[U,V]=0$ is equivalent to the new two-component Sasa--Satsuma equation \eqref{SS}. On the other hand, it is noted the matrices $U$ and $V$ obey the symmetry conditions:
\begin{align}
U(x,t;\xi)&=-U^\dag(x,t;\xi^*),\quad V(x,t;\xi)=-V^\dag(x,t;\xi^*),\label{2.5}\\
U(x,t;\xi)&=\mathcal{A}U^*(x,t;-\xi^*)\mathcal{A}, \quad V(x,t;\xi)=\mathcal{A}V^*(x,t;-\xi^*)\mathcal{A},\label{2.6}
\end{align}
where
\be
\mathcal{A}=\begin{pmatrix}1 & \mathbf{0}_{1\times3} \\ \mathbf{0}_{3\times1} & \sigma_1
\end{pmatrix},\quad \sigma_1=\begin{pmatrix} 0&1&0\\1&0&0\\0&0&1\end{pmatrix}.
\ee

Introducing a new eigenfunction $\mu(x,t;\xi)$ by
\begin{equation}
\mu(x,t;\xi)=\psi(x,t;\xi) \e^{\i(\xi x-4\xi^3t)\sigma},
\end{equation}
we obtain
\begin{align}
&\mu_x=-\ii\xi[\sigma,\mu]+U_1\mu,\label{2.9}\\
&\mu_t=4\ii\xi^3[\sigma,\mu]+V_1\mu,\label{2.10}
\end{align}
We now define two Jost solutions $\mu_\pm(x,t;\xi)$ of \eqref{2.9} for $\xi\in\bfR$ with $\mu_\pm(x,t;\xi)\to\mathbb{I}_{4\times4}$ as $x\to\pm\infty$ by the following Volterra integral equations
\be
\mu_\pm(x,t;\xi)=\mathbb{I}_{4\times4}+\int_{\pm\infty}^x\e^{\ii\xi(y-x)\sigma}[U_1(y,t)\mu_\pm(y,t;\xi)]
\e^{-\ii\xi(y-x)\sigma}\dd y.\label{2.11}
\ee
Denote $\mu_{\pm L}(x,t;\xi)$ and $\mu_{\pm R}(x,t;\xi))$ be the first column and last three columns of the $4\times4$ matrices $\mu_\pm(x,t;\xi)$. Then, the following are consequences of standard analysis of the iterates that:

$(i)$ $\mu_{+L}$, $\mu_{-R}$ are analytic in $\bfC_+$ and can be continuously extended to $\bfC_+\cup\bfR$, $(\mu_{+L},\mu_{-R})\rightarrow \mathbb{I}_{4\times4}$ as $\xi\rightarrow\infty$;

$(ii)$ $\mu_{-L}$, $\mu_{+R}$ are analytic in $\bfC_-$ and can be continuously extended to $\bfC_-\cup\bfR$, $(\mu_{-L},\mu_{+R})\rightarrow \mathbb{I}_{4\times4}$ as $\xi\rightarrow\infty$.

Since $\mu_\pm(x,t;\xi)$ are both fundamental matrices of solutions of \eqref{2.9}, thus, they satisfy the scattering relation
\be\label{2.12}
\mu_+(x,t;\xi)=\mu_-(x,t;\xi)\e^{-\ii(\xi x-4\xi^3t)\sigma}s(\xi)\e^{\ii(\xi x-4\xi^3t)\sigma},\quad \xi\in\bfR.
\ee
Evaluation at $x\rightarrow-\infty,\ t=0$ gives
\be\label{2.13}
s(\xi)=\lim_{x\rightarrow-\infty}\e^{\ii\xi x\sigma}\mu_+(x,0;\xi)\e^{-\ii\xi x\sigma},
\ee
that is,
\be\label{2.14}
s(\xi)=\mathbb{I}_{4\times4}-\int_{-\infty}^{+\infty}\e^{\ii\xi x\sigma}[U(x,0)\mu_+(x,0;\xi)]\e^{-\ii\xi x\sigma}\dd x.
\ee
This implies that the scattering matrix $s(\xi)$ can be determined in terms of the initial values $u_0(x)$ and $w_0(x)$.

The symmetries in \eqref{2.5} and \eqref{2.6} implies that
\be\label{2.15}
\mu_\pm^{-1}(x,t;\xi)=\mu_\pm^\dag(x,t;\xi^*),\quad \mu_\pm(x,t;\xi)=\mathcal{A}\mu_\pm^*(x,t;-\xi^*)\mathcal{A}.
\ee
Moreover, the tracelessness of $U(x,t)$ shows that $\det[\mu_\pm(x,t;\xi)]=1.$ Then \eqref{2.12} yields $\det[s(\xi)]=1.$
By \eqref{2.15}, the $4\times4$ matrix-valued spectral function $s(k)$ obeys the symmetries
\be\label{2.16}
s^{-1}(\xi)=s^\dag(\xi^*),\quad s(\xi)=\mathcal{A}s^*(-\xi^*)\mathcal{A},\quad \xi\in\bfR.
\ee
It follows from the first symmetry in \eqref{2.16} that we can write $s(\xi)$ as
\be\label{2.17}
s(\xi)=\begin{pmatrix}
\det[a^\dag(\xi^*)] & b(\xi)\\
-\text{adj}[a^\dag(\xi^*)]b^\dag(\xi^*) & a(\xi)
\end{pmatrix},\quad\xi\in\bfR,
\ee
where adj$(A)$ denotes the adjoint matrix of $A$ in the context of linear algebra, $a(\xi)$ is a $3\times3$ matrix-valued function and $b(\xi)$ is a $1\times3$ row vector-valued function. On the other hand, it is easy to see from \eqref{2.14} that $a(\xi)$ is analytic for $\xi\in\bfC_-$, however, $b(\xi)$ is only defined in $\bfR$. Furthermore, we also from the second symmetry in \eqref{2.16} have
\be
a(\xi)=\sigma_1a^*(-\xi^*)\sigma_1,\quad \xi\in\bfC_-;\quad b(\xi)=b^*(-\xi^*)\sigma_1,\quad \xi\in\bfR.
\ee

To exclude soliton-type phenomena, for convenience, we assume that $\det[a(\xi)]$ has no zeros in $\bfC_-$. Then, we have the following main result in this section, which shows how solutions of \eqref{SS} can be constructed by a basic $4\times4$ matrix Riemann--Hilbert problem.
\begin{theorem}\label{th2.1}
Define a piecewise meromorphic $4\times4$ matrix-valued function
\be
M(x,t;\xi)=\left\{
\begin{aligned}
&\left(\mu_{+L}(x,t;\xi)\det[a^\dag(\xi^*)]^{-1},\mu_{-R}(x,t;\xi)\right),\quad \xi\in\bfC_+,\\
&\left(\mu_{-L}(x,t;\xi),\mu_{+R}(x,t;\xi)a^{-1}(\xi)\right),\qquad\qquad \xi\in\bfC_-,
\end{aligned}
\right.
\ee
and the $4\times4$ matrix-valued jump matrix $J(x,t;\xi)$ by
\be\label{2.20}
J(x,t;\xi)=\begin{pmatrix}
1+\gamma(\xi)\gamma^\dag(\xi) & \gamma(\xi)\e^{2\ii(\xi x-4\xi^3t)}\\
\gamma^\dag(\xi)\e^{-2\ii(\xi x-4\xi^3t)} & \mathbb{I}_{3\times3}
\end{pmatrix},
\ee
where
\be\label{2.21}
\gamma(\xi)=-b(\xi)a^{-1}(\xi),\quad \gamma(\xi)=\gamma^*(-\xi^*)\sigma_1, \quad \xi\in\bfR.
\ee
Then the following $4\times4$ matrix RH problem:\\
$\bullet$ $M(x,t;\xi)$ is a sectionally meromorphic function with respect to $\bfR$;\\
$\bullet$ The limiting values $M_\pm(x,t;\xi)=\lim_{\varepsilon\to0}M(x,t;\xi\pm\ii\varepsilon)$ satisfy the jump condition $M_+(x,t;\xi)=M_-(x,t;\xi)J(x,t;\xi)$ for $\xi\in\bfR$;\\
$\bullet$ As $\xi\rightarrow\infty$, $M(x,t;\xi)$ has the asymptotics: $M(x,t;\xi)=\mathbb{I}_{4\times4}+O\left(\xi^{-1}\right)$;\\
has a unique solution for each $(x,t)\in\bfR^2$.

Moreover, define $\{u(x,t),w(x,t)\}$ in terms of $M(x,t;\xi)$ by
\begin{align}
\begin{pmatrix}u(x,t) & u^*(x,t) & w(x,t)\end{pmatrix}=2\ii\lim_{\xi\rightarrow\infty}(\xi M(x,t;\xi))_{12},\label{2.22}
\end{align}
which solves the Cauchy problem of new two-component Sasa--Satsuma equation \eqref{SS}.
\end{theorem}
\begin{proof}
It is a simple consequence of Liouville's theorem that if a solution exists, it is unique. The existence of solution of RH problem follows by means of Zhou's vanishing lemma argument \cite{ZX} since
\be\label{2.24}
J(x,t;\xi)=J^\dag(x,t;\xi^*)=\mathcal{A}J^*(x,t;-\xi^*)\mathcal{A},\quad \xi\in\bfR.
\ee
Expanding this solution as $\xi\to\infty$,
\be\label{2.25}
M(x,t;\xi)=\mathbb{I}_{4\times4}+\frac{M_1(x,t)}{\xi}+\frac{M_2(x,t)}{\xi^2}+O(\xi^{-3})
\ee
and inserting this into equation \eqref{2.9} one finds that the solution of \eqref{SS} is given by \eqref{2.22}.

\end{proof}

\section{Asymptotic analysis in oscillating sector $\mathcal{O}$}\label{sec3}
The representation of the solution of the Cauchy problem for a nonlinear integrable equation in terms of the solution of an associated RH problem makes it possible to analyze the long-time asymptotics via the Deift--Zhou steepest descent method \cite{PD}. The goal of this section is devoted to deriving the long-time asymptotic behavior of solution for the new two-component Sasa--Satsuma equation \eqref{SS} in oscillating sector $\mathcal{O}$ defined by
\be\label{3.1}
\mathcal O=\{(x,t)\in\bfR^2|t>1,\ 0<x\leq Nt,\ x^3/t\to\infty\},\, N \,\text{constant}.
\ee

We first prove an important result (Theorem \ref{th3.1}), which expresses the large $z$ behavior of solution of a model RH problem in terms of the solution of parabolic cylinder functions and will be very useful in the study of long-time asymptotics in the oscillating sector.
\subsection{A model RH problem}\label{sec3.1}
\begin{figure}[htbp]
  \centering
  \includegraphics[width=3.5in]{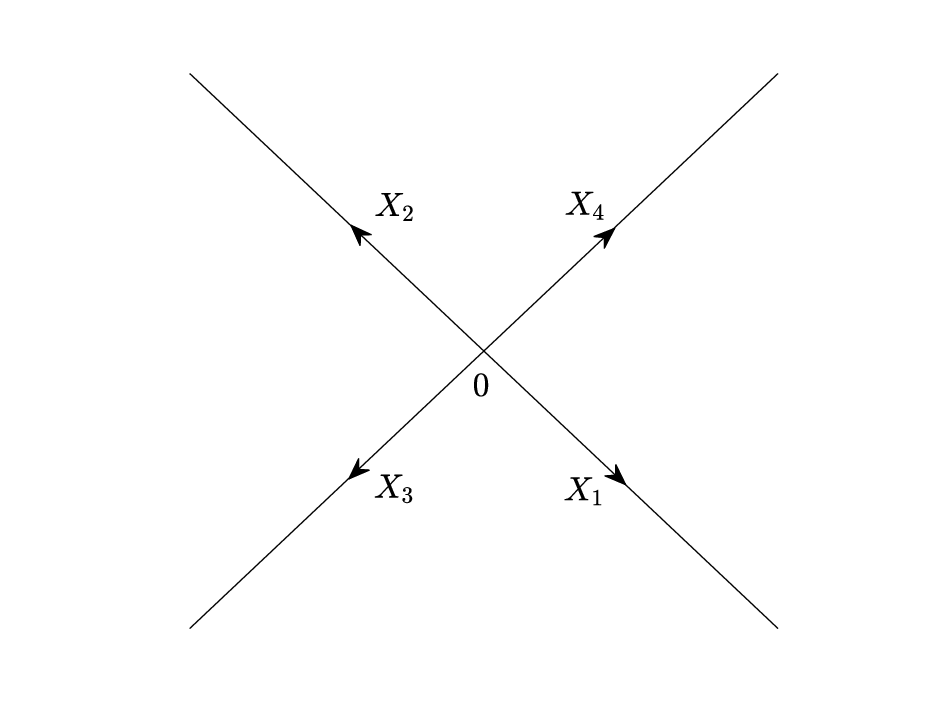}
  \caption{The oriented contour $X=\cup_{j=1}^4X_j$.}\label{fig1}
\end{figure}
Define the oriented contour $X=\cup_{j=1}^4X_j$  by
\be\label{3.2}
\begin{aligned}
X_1&=\{\epsilon\e^{-\frac{\ii\pi}{4}}|0\leq\epsilon<\infty\},\,\,\,\,\ X_2=\{\epsilon\e^{\frac{3\ii\pi}{4}}|0\leq\epsilon<\infty\},\\
X_3&=\{\epsilon\e^{-\frac{3\ii\pi}{4}}|0\leq \epsilon<\infty\},\,\,\
 X_4=\{\epsilon\e^{\frac{\ii\pi}{4}}|0\leq\epsilon<\infty\},
\end{aligned}
\ee
see Figure \ref{fig1}. For a $1\times3$ complex-valued row vector $\rho$, define the function $\nu$ by $\nu(\rho)=\frac{1}{2\pi}\ln(1+\rho\rho^\dag)>0$ and the $4\times 4$ jump matrix by
\be\label{3.3}
J^X(\rho;z)=\left\{
\begin{aligned}
&\begin{pmatrix}
1 & \rho\e^{-\frac{\ii z^2}{2}}z^{-2\ii\nu(\rho)}\\
\textbf{0}_{3\times1} & \mathbb{I}_{3\times3}
\end{pmatrix},\qquad\quad\,\ z\in X_1,\\
&\begin{pmatrix}
1 & -\frac{\rho}{1+\rho\rho^\dag}\e^{-\frac{\ii z^2}{2}}z^{-2\ii\nu(\rho)}\\
\textbf{0}_{3\times1} & \mathbb{I}_{3\times3}
\end{pmatrix}, \,\,\,\ z\in X_2,\\
&\begin{pmatrix}
1 & \textbf{0}_{1\times3}\\
-\frac{\rho^\dag}{1+\rho\rho^\dag}\e^{\frac{\ii z^2}{2}}z^{2\ii\nu(\rho)} & \mathbb{I}_{3\times3}
\end{pmatrix},\qquad z\in X_3,\\
&\begin{pmatrix}
1 & \textbf{0}_{1\times3}\\
\rho^\dag\e^{\frac{\ii z^2}{2}}z^{2\ii\nu(\rho)} & \mathbb{I}_{3\times3}
\end{pmatrix},\qquad\qquad z\in X_4.
\end{aligned}
\right.
\ee
We consider the following RH problem:\\
$\bullet$ $M^X(\rho;z)$ is analytic for $z\in\bfC\setminus X$ and extends continuously to $X$;\\
$\bullet$ Across $X$, the boundary values $M^X_\pm$ satisfy the jump relation $M^X_+(\rho;z)=M^X_-(\rho;z)J^X(\rho;z)$;\\
$\bullet$ $M^X(\rho;z)\rightarrow \mathbb{I}_{4\times4}$, as $z\rightarrow\infty$.

\begin{theorem}\label{th3.1}
The RH problem has a unique solution $M^X(\rho;z)$ for each $1\times3$ row vector $\rho$, and this solution satisfies
\be\label{3.4}
M^X(\rho;z)=\mathbb{I}_{4\times4}+\frac{M^X_1(\rho)}{z}+O\left(z^{-2}\right),\quad z\rightarrow\infty,
\ee
where
\be\label{3.5}
\left(M^X_1\right)_{12}(\rho)=-\ii\beta^X,\, \left(M^X_1\right)_{21}(\rho)=-\ii\left(\beta^X\right)^\dag,\,
\beta^X=\frac{\nu\Gamma(-\ii\nu)\e^{\frac{\ii\pi}{4}-\frac{\pi\nu}{2}}}{\sqrt{2\pi}}\rho,
\ee
where $\Gamma(\cdot)$ denotes the standard Gamma function. Moreover, for each compact subset $\mathcal{D}$ of $\bfC$,
\be\label{3.6}
\sup_{\rho\in\mathcal{D}}\sup_{z\in\bfC\setminus X}|M^X(\rho;z)|<\infty.
\ee
\end{theorem}
\begin{proof}
See Appendix \ref{secA}.
\end{proof}

\subsection{Transformations}
The Deift--Zhou nonlinear steepest descent method for RH problems consists of making a series of invertible transformations in order to arrive at a problem that can be approximated in the large-$t$ limit. For this purpose, we first note that the jump matrix $J(x,t;\xi)$ defined in \eqref{2.20} involves the exponentials $\e^{\pm t\Phi(x,t;\xi)}$, where $\Phi(x,t;\xi)$ is given by
\be
\Phi(x,t;\xi)=8\ii\xi^3-2\ii\xi\frac{x}{t}.
\ee
Suppose $(x,t)\in\mathcal{O}$. By solving the equation $\partial\Phi(x,t;\xi)/\partial \xi=0,$ we see that there are two real critical points located at
\be\label{3.8}
\pm \xi_0=\pm\sqrt{\frac{x}{12t}},
\ee
moreover, the signature table for Re$\Phi(x,t;\xi)$ is shown in Figure \ref{fig2}.
\begin{figure}[htbp]
  \centering
  \includegraphics[width=3.5in]{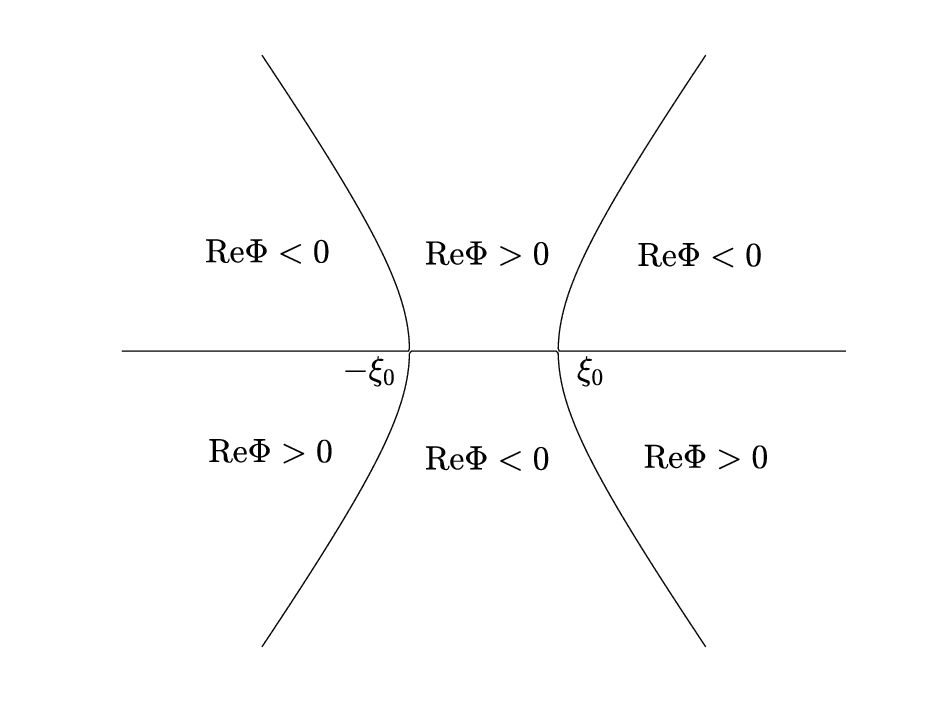}
  \caption{The signature table for Re$\Phi(x,t;\xi)$ and $\pm\xi_0$.}\label{fig2}
\end{figure}

The jump matrix $J(x,t;\xi)$ enjoys two distinct factorizations:
\be\label{3.9}
J(x,t;\xi)=\left\{
\begin{aligned}
&\begin{pmatrix}
1& \gamma\e^{-t\Phi}\\
\mathbf{0}_{3\times1} & \mathbb{I}_{3\times3}
\end{pmatrix}
 \begin{pmatrix}
1 & \mathbf{0}_{1\times3} \\
\gamma^\dag\e^{t\Phi} & \mathbb{I}_{3\times3}
\end{pmatrix},\\
&\begin{pmatrix}
1 & \mathbf{0}_{1\times3}\\
\frac{\gamma^\dag\e^{t\Phi}}{1+\gamma\gamma^\dag} & \mathbb{I}_{3\times3}
\end{pmatrix}
\begin{pmatrix}
1+\gamma\gamma^\dag & \mathbf{0}_{1\times3}\\
\mathbf{0}_{3\times1} & (\mathbb{I}_{3\times3}+\gamma^\dag\gamma)^{-1}
\end{pmatrix}
\begin{pmatrix}
1 & \frac{\gamma\e^{-t\Phi}}{1+\gamma\gamma^\dag} \\
\mathbf{0}_{3\times1} & \mathbb{I}_{3\times3}
\end{pmatrix}.
\end{aligned}
\right.
\ee
Thus, by signature table in Figure \ref{fig2}, we can see that the jump matrix has the wrong factorization for $\xi\in(-\xi_0,\xi_0)$. Hence the first transformation is to introduce $M^{(1)}$ by
\be\label{3.10}
M^{(1)}(x,t;\xi)=M(x,t;\xi)\begin{pmatrix}[\det\delta(\xi)]^{-1} & \mathbf{0}_{1\times3}\\
\mathbf{0}_{3\times1} & \delta(\xi)\end{pmatrix},
\ee
where the $3\times3$ matrix-valued function satisfies the following RH problem:\\
$\bullet$ $\delta(\xi)$ is analytic for $\xi\in\bfC\setminus[-\xi_0,\xi_0]$, and it takes continuous boundary values on $(-\xi_0,\xi_0)$ from the upper and lower half-planes;\\
$\bullet$ The boundary values on the jump contour $(-\xi_0,\xi_0)$ (oriented to right) are related as
\be
\delta_+(\xi)=(\mathbb{I}_{3\times3}+\gamma^\dag(\xi)\gamma(\xi))\delta_-(\xi), \quad \xi\in(-\xi_0,\xi_0);
\ee
$\bullet$ $\delta(\xi)=\mathbb{I}_{3\times3}+O(\xi^{-1})$ as $\xi\to\infty$.\\
Since the jump matrix $\mathbb{I}_{3\times3}+\gamma^\dag(\xi)\gamma(\xi)$ is positive definite, the vanishing lemma of Zhou \cite{ZX} yields the existence and uniqueness of the function $\delta(\xi)$. Furthermore, it is easy to see that $\det\delta(\xi)$ obeys a scalar RH problem:\\
$\bullet$ $\det\delta(\xi)$ is analytic for $\xi\in\bfC\setminus[-\xi_0,\xi_0]$;\\
$\bullet$ $\det\delta(\xi)$ takes continuous boundary values $\det\delta_\pm(\xi)$ on $(-\xi_0,\xi_0)$, and they are related by the jump condition
\be
\det\delta_+(\xi)=\det\delta_-(\xi)(1+\gamma(\xi)\gamma^\dag(\xi)), \quad \xi\in(-\xi_0,\xi_0);
\ee
$\bullet$ $\lim_{\xi\to\infty}\det\delta(\xi)=1$.
\begin{proposition}
The functions $\delta(\xi)$ and $\det\delta(\xi)$ have the following properties:
\be
\delta(\xi)=[\delta^\dag(\xi^*)]^{-1}=\sigma_1\delta^*(-\xi^*)\sigma_1,\quad
\det\delta(\xi)=[(\det\delta(\xi^*))^*]^{-1}=[\det\delta(-\xi^*)]^*,
\ee
and
\be
|\delta(\xi)|,|\det\delta(\xi)|\leq\text{const}<\infty,
\ee
where $|X|^2=\text{tr}(X^\dag X)$ for any matrix $X$.
\end{proposition}
Then $M^{(1)}$ satisfies the following matrix RH problem:\\
$\bullet$ $M^{(1)}(x,t;\xi)$ is analytic for $\xi\in\bfC\setminus\bfR$, and it takes continuous boundary values on $\bfR$ from the upper and lower half-planes;\\
$\bullet$ The boundary values $M^{(1)}_\pm(x,t;\xi)$ on the jump contour $\bfR$ are related as
\be
M^{(1)}_+(x,t;\xi)=M^{(1)}_-(x,t;\xi)J^{(1)}(x,t;\xi), \quad \xi\in\bfR,
\ee
where the jump matrix is given by
\begin{align}
J^{(1)}(x,t;\xi)=\left\{
\begin{aligned}
&\begin{pmatrix}
1 & \gamma\delta\det\delta\e^{-t\Phi}\\
\mathbf{0}_{3\times1} & \mathbb{I}_{3\times3}
\end{pmatrix}
\begin{pmatrix}
1 &  \mathbf{0}_{1\times3}\\
\delta^{-1}\gamma^\dag(\det\delta)^{-1}\e^{t\Phi} & \mathbb{I}_{3\times3}
\end{pmatrix},\qquad\qquad\quad\,\ |\xi|>\xi_0,\\
&\begin{pmatrix}
1 & \mathbf{0}_{1\times3}\\
\delta_-^{-1}\frac{\gamma^\dag}{1+\gamma\gamma^\dag}(\det\delta_-)^{-1}\e^{t\Phi} & \mathbb{I}_{3\times3} \\
\end{pmatrix}
\begin{pmatrix}
1 & \delta_+\frac{\gamma}{1+\gamma\gamma^\dag}\det\delta_+\e^{-t\Phi} \\
\mathbf{0}_{3\times1} & \mathbb{I}_{3\times3}
\end{pmatrix},\ |\xi|<\xi_0;
\end{aligned}
\right.
\end{align}
$\bullet$ $M^{(1)}(x,t;\xi)=\mathbb{I}_{4\times4}+O(\xi^{-1})$ as $\xi\to\infty$.

\begin{figure}[htbp]
  \centering
  \includegraphics[width=3.5in]{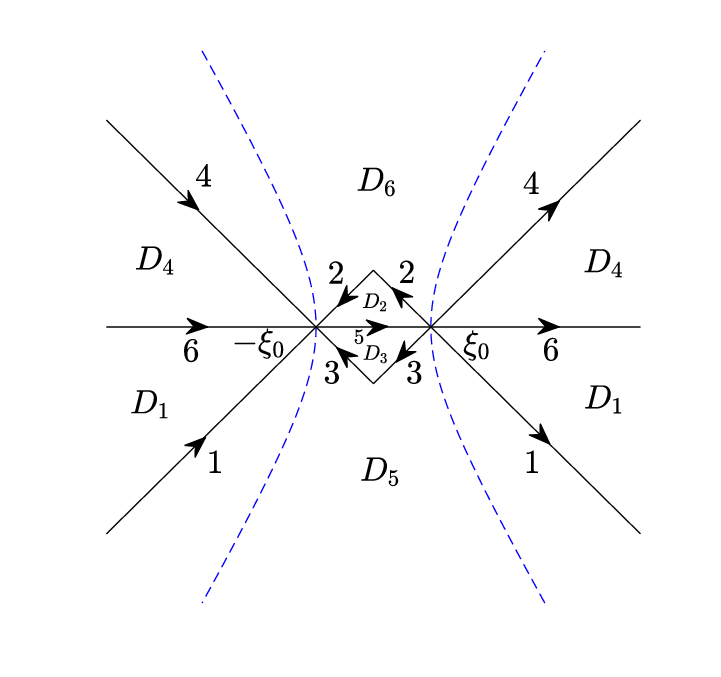}
  \caption{The jump contour $\Sigma$ (solid line) and the open sets $\{D_j\}_1^6$.}\label{fig3}
\end{figure}
The next step is to deform the contour such that the jump matrix involves the exponential factor $\e^{t\Phi}$ on the parts of the contour where Re$\Phi$ is negative, the factor $\e^{-t\Phi}$ on the parts where Re$\Phi$ is positive and the jumps on the original contour $\bfR$ are small remainders with respect to $t$. To achieve this goal, we should introduce the analytic approximations of $\gamma(\xi)$ and $\frac{\gamma^\dag(\xi)}{1+\gamma(\xi)\gamma^\dag(\xi)}$. The symmetry of $\gamma(\xi)$ in \eqref{2.21} implies that we can rewrite
\begin{align}
\gamma(\xi)=&\begin{pmatrix}\gamma_1(\xi) & \gamma_1^*(-\xi^*) & \gamma_2(\xi)\end{pmatrix},\ \xi\in\bfR,\label{3.17}\\
\tilde{\gamma}^\dag(\xi)\doteq&\frac{\gamma^\dag(\xi)}{1+\gamma(\xi)\gamma^\dag(\xi)}
=\begin{pmatrix}\gamma_3(\xi) & \gamma_3^*(-\xi^*) & \gamma_4(\xi)\end{pmatrix},\ \xi\in\bfR.
\end{align}
Let $D_j\doteq D_j(\xi)$, $j=1,\cdots,6$ denote the open subsets of $\bfC$ displayed in Figure \ref{fig3}, then we have the following lemma.
\begin{lemma}\label{lem3.1}
There exist decompositions
\begin{align}
\gamma_1(\xi)=&\gamma_{1,a}(x,t;\xi)+\gamma_{1,r}(x,t;\xi),\quad |\xi|>\xi_0,~\xi\in\bfR,\label{3.19}\\
\gamma_3(\xi)=&\gamma_{3,a}(x,t;\xi)+\gamma_{3,r}(x,t;\xi),\quad |\xi|<\xi_0,~\xi\in\bfR,\label{3.20}
\end{align}
where the functions $\gamma_{j,a}$ and $\gamma_{j,r}$, $j=1,3$ have the following properties:

(1) For each $(x,t)\in\mathcal{O}$ and $j=1,3$, $\gamma_{j,a}(x,t;\xi)$is defined and continuous for $\xi\in\bar{D}_j$ and
analytic in $D_j$.

(2) The function $\gamma_{j,a}$ satisfies the following estimates
\be\label{3.21}
|\gamma_{j,a}(x,t;\xi)-\gamma_j(\xi_0)|\leq C|\xi-\xi_0|\e^{\frac{t}{4}|\text{Re}\Phi(x,t;\xi)|},~\xi\in\bar{D}_j,\ j=1,3,
\ee
and
\be\label{3.22}
|\gamma_{1,a}(x,t;\xi)|\leq \frac{C}{1+|\xi|^2}\e^{\frac{t}{4}|\text{Re}\Phi(x,t;\xi)|},
~\xi\in\bar{D}_1.
\ee

(3) The $L^1, L^2$ and $L^\infty$ norms of the function $\gamma_{1,r}(x,t;\cdot)$ on $(-\infty,-\xi_0)\cup(\xi_0,\infty)$ are $O(t^{-3/2})$ as $t\rightarrow\infty$ uniformly with respect to $(x,t)\in\mathcal{O}$.

(4) The $L^1, L^2$ and $L^\infty$ norms of the function $\gamma_{3,r}(x,t;\cdot)$ on $(-\xi_0,\xi_0)$ are $O(t^{-3/2})$ as $t\rightarrow\infty$ uniformly with respect to $(x,t)\in\mathcal{O}$.

\end{lemma}
\begin{proof}
See Lemma 4.8 in \cite{Lenells-IUMJ}.
\end{proof}
For $l=2,4$, the decomposition of $\gamma_l(\xi)=\gamma_{l,a}(x,t;\xi)+\gamma_{l,r}(x,t;\xi)$ can be similarly found. Thus, we establish the decompositions $\gamma(\xi)=\gamma_a(x,t;\xi)+\gamma_r(x,t;\xi)$ of $\gamma(\xi)$, and $\tilde{\gamma}^\dag(\xi)=\tilde{\gamma}^\dag_a(x,t;\xi)+\tilde{\gamma}^\dag_r(x,t;\xi)$ of $\tilde{\gamma}^\dag(\xi)$ by setting
\begin{align*}
\gamma_a(x,t;\xi)=&\begin{pmatrix} \gamma_{1,a}(x,t;\xi) & \gamma_{1,a}^*(x,t;-\xi^*) & \gamma_{2,a}(x,t;\xi)\end{pmatrix},\\
\gamma_r(x,t;\xi)=&\begin{pmatrix} \gamma_{1,r}(x,t;\xi) & \gamma_{1,r}^*(x,t;-\xi^*) & \gamma_{2,r}(x,t;\xi)\end{pmatrix},\\
\tilde{\gamma}_a^\dag(x,t;\xi)=&\begin{pmatrix} \tilde{\gamma}_{3,a}(x,t;\xi) & \tilde{\gamma}_{3,a}^*(x,t;-\xi^*) & \tilde{\gamma}_{4,a}(x,t;\xi)\end{pmatrix},\\
\tilde{\gamma}^\dag_r(x,t;\xi)=&\begin{pmatrix} \tilde{\gamma}_{3,r}(x,t;\xi) & \tilde{\gamma}_{3,r}^*(x,t;-\xi^*) & \tilde{\gamma}_{4,r}(x,t;\xi)\end{pmatrix}.
\end{align*}

Now, we can introduce $M^{(2)}(x,t;\xi)$ by
\be\label{3.23}
M^{(2)}(x,t;\xi)=M^{(1)}(x,t;\xi)G(\xi),
\ee
where the sectionally analytic function $G$ is defined by
\be\label{3.24}
G(\xi)=\left\{
\begin{aligned}
&\begin{pmatrix}
1 & \gamma_{a}\delta\det\delta\e^{-t\Phi}\\
\mathbf{0}_{3\times1} & \mathbb{I}_{3\times3}
\end{pmatrix},\qquad\quad\,\,\ \xi\in D_1,\\
&\begin{pmatrix}
1 & -\delta_+\tilde{\gamma}_{a}\det\delta_+\e^{-t\Phi}\\
\mathbf{0}_{3\times1} & \mathbb{I}_{3\times3}
\end{pmatrix},\quad\,\,\,\ \xi\in D_2,\\
&\begin{pmatrix}
1 & \mathbf{0}_{1\times3}\\
\delta_-^{-1}\tilde{\gamma}_{a}^\dag(\det\delta_-)^{-1}\e^{t\Phi} & \mathbb{I}_{3\times3}
\end{pmatrix},\quad\ \xi\in D_3,\\
&\begin{pmatrix}
1 &  \mathbf{0}_{1\times3}\\
-\delta^{-1}\gamma^\dag_{a}(\det\delta)^{-1}\e^{t\Phi} & \mathbb{I}_{3\times3}
\end{pmatrix},\quad \xi\in D_4,\\
&\mathbb{I}_{4\times4}, \qquad\qquad\qquad\qquad\qquad\qquad\,\,\ \xi\in D_5\cup D_6.
\end{aligned}
\right.
\ee
Let $\Sigma\subset\bfC$ denote the contour displayed in Figure \ref{fig3}. It follows that $M^{(2)}(x,t;\xi)$ satisfies the following RH problem:\\
$\bullet$ $M^{(2)}(x,t;\xi)$ is analytic off $\Sigma$, and it takes continuous boundary values on $\Sigma$;\\
$\bullet$ Across the oriented contour $\Sigma$, the boundary values $M^{(2)}_\pm(x,t;\xi)$ are connected by the following formula:
\be
M^{(2)}_+(x,t;\xi)=M^{(2)}_-(x,t;\xi)J^{(2)}(x,t;\xi), \quad \xi\in\Sigma;
\ee
$\bullet$ $M^{(2)}(x,t;\xi)\to\mathbb{I}_{4\times4}$, as $\xi\to\infty$;\\
where the jump matrix $J^{(2)}(x,t;\xi)$ is given by
\be
\begin{aligned}
J^{(2)}_1=&\begin{pmatrix}
1 & \gamma_{a}\delta\det\delta\e^{-t\Phi}\\
\mathbf{0}_{3\times1} & \mathbb{I}_{3\times3}
\end{pmatrix},\
J^{(2)}_2=\begin{pmatrix}
1 & -\delta_+\tilde{\gamma}_{a}\det\delta_+\e^{-t\Phi}\\
\mathbf{0}_{3\times1} & \mathbb{I}_{3\times3}
\end{pmatrix},\\
J^{(2)}_3=&\begin{pmatrix}
1 & \mathbf{0}_{1\times3}\\
-\delta_-^{-1}\tilde{\gamma}_{a}^\dag(\det\delta_-)^{-1}\e^{t\Phi} & \mathbb{I}_{3\times3}
\end{pmatrix},\
J^{(2)}_4=\begin{pmatrix}
1 &  \mathbf{0}_{1\times3}\\
\delta^{-1}\gamma^\dag_{a}(\det\delta)^{-1}\e^{t\Phi} & \mathbb{I}_{3\times3}
\end{pmatrix},\\
J^{(2)}_5=&\begin{pmatrix}
1 & \mathbf{0}_{1\times3}\\
\delta_-^{-1}\tilde{\gamma}_r^\dag(\det\delta_-)^{-1}\e^{t\Phi} & \mathbb{I}_{3\times3} \\
\end{pmatrix}
\begin{pmatrix}
1 & \delta_+\tilde{\gamma}_r\det\delta_+\e^{-t\Phi} \\
\mathbf{0}_{3\times1} & \mathbb{I}_{3\times3}
\end{pmatrix},\\
J^{(2)}_6=&\begin{pmatrix}
1 & \gamma_r\delta\det\delta\e^{-t\Phi}\\
\mathbf{0}_{3\times1} & \mathbb{I}_{3\times3}
\end{pmatrix}
\begin{pmatrix}
1 &  \mathbf{0}_{1\times3}\\
\delta^{-1}\gamma^\dag_r(\det\delta)^{-1}\e^{t\Phi} & \mathbb{I}_{3\times3}
\end{pmatrix},
\end{aligned}
\ee
with $J^{(2)}_j$ denoting the restriction of $J^{(2)}$ to the contour labeled by $j$ in Figure \ref{fig3}.

\subsection{Local models}
Obviously, as $t\rightarrow\infty$, the matrix $J^{(2)}-\mathbb{I}_{4\times4}$ decays to zero everywhere except near the critical points $\pm \xi_0$. This implies that the main contribution to the long-time asymptotics of $M^{(2)}$ should come from the neighborhoods of $\pm\xi_0$. 

In order to relate $M^{(2)}$ to the solution $M^X$ of the model RH problem in Theorem \ref{th3.1}, we introduce the following scaling transforms for $\xi$ near $\pm\xi_0$
\begin{align}
&T_{-\xi_0}:~\xi\mapsto\frac{z}{\sqrt{48t\xi_0}}-\xi_0,\label{3.27}\\
&T_{\xi_0}:~\xi\mapsto\frac{z}{\sqrt{48t\xi_0}}+\xi_0.\label{3.28}
\end{align}
Observe that the $3\times3$ matrix-valued function $\delta(\xi)$ can not be expressed explicitly, in order to proceed to the next step, we write, for example,
\be
\left(\gamma_{a}\delta\det\delta\e^{-t\Phi}\right)(\xi)=\left(\gamma_{a}(\delta-\det\delta \mathbb{I}_{3\times3})\det\delta\e^{-t\Phi}\right)(\xi)
+\left(\gamma_{a}(\det\delta)^2\e^{-t\Phi}\right)(\xi).\label{3.29}
\ee
For the second term in the right-hand side of \eqref{3.29}, by Plemelj formula, we know that
\be\label{3.30}
\det\delta(\xi)=\exp\left\{\frac{1}{2\pi\ii}\int_{-\xi_0}^{\xi_0}\frac{\ln(1+\gamma(s)\gamma^\dag(s))}{s-\xi}\dd s\right\}
=\left(\frac{\xi-\xi_0}{\xi+\xi_0}\right)^{-\ii\nu(\xi_0)}\e^{\chi(\xi)},
\ee
where
\begin{align}
\nu(\xi_0)=&\frac{1}{2\pi}\ln(1+\gamma(\xi_0)\gamma^\dag(\xi_0))>0,\label{3.31}\\
\chi(\xi)=&\frac{1}{2\pi\ii}\int_{-\xi_0}^{\xi_0}\ln\left(\frac{1+\gamma(s)\gamma^\dag(s)}
{1+\gamma(\xi_0)\gamma^\dag(\xi_0)}\right)\frac{\dd s}{s-\xi},\label{3.32}
\end{align}
thus, a direct computation yields that
\bea
T_{\xi_0}\left(\gamma_{a}(\det\delta)^2\e^{-t\Phi}\right)(\xi)
=\vartheta^2\upsilon^2\gamma_{a}\left(\frac{z}{\sqrt{48t\xi_0}}+\xi_0\right),\label{3.33}
\eea
where
\begin{align}
\vartheta&=(192t\xi_0^3)^\frac{\ii\nu}{2}\e^{8\ii t\xi_0^3+\chi(\xi_0)},\label{3.34}\\
\upsilon&=z^{-\ii\nu}\e^{-\frac{\ii z^2}{4}(1+ z/\sqrt{432t\xi_0^3})}
\left(\frac{2\xi_0}{z/\sqrt{48t\xi_0}+2\xi_0}\right)^{-\ii\nu}\e^{\chi([z/\sqrt{48t\xi_0}]+\xi_0)-\chi(\xi_0)}.\label{3.35}
\end{align}
However, for the first term in the right-hand side of \eqref{3.29}, if we denote
\be
\tilde{\delta}(\xi)=\left(\gamma_{a}(\delta-\det\delta\mathbb{I}_{3\times3})\e^{-t\Phi}\right)(\xi),
\ee
one can find that $\tilde{\delta}$ satisfies the following RH problem:\\
$\bullet$ $\tilde{\delta}(\xi)$ is analytic in $\xi\in\bfC\setminus[-\xi_0,\xi_0]$ with continuous boundary values on $(-\xi_0,\xi_0)$;\\
$\bullet$ On the jump contour $(-\xi_0,\xi_0)$, $\tilde{\delta}(\xi)$ satisfies the jump condition
\be
\tilde{\delta}_+(\xi)=(1+\gamma(\xi)\gamma^\dag(\xi))\tilde{\delta}_-(\xi)+f(\xi)\e^{-t\Phi(x,t;\xi)}, \quad \xi\in(-\xi_0,\xi_0),
\ee
where
\be
f(\xi)=\gamma_{a}(\xi)[\gamma^\dag(\xi)\gamma(\xi)-\gamma(\xi)\gamma^\dag(\xi)\mathbb{I}_{3\times3}]\delta_-(\xi);
\ee
$\bullet$ $\tilde{\delta}(\xi)\to\mathbf{0}_{3\times3},$ as $\xi\to\infty$.\\
By \cite{MJA-ASF}, the function $\tilde{\delta}(\xi)$ can be expressed by
\be\label{3.39}
\begin{aligned}
\tilde{\delta}(\xi)&=X(\xi)\int_{-\xi_0}^{\xi_0}\frac{\e^{-t\Phi(x,t;s)}f(s)}{X_+(s)(s-\xi)}\dd s,\\
X(\xi)&=\exp\left\{\frac{1}{2\pi\ii}\int_{-\xi_0}^{\xi_0}\frac{\ln(1+\gamma(s)\gamma^\dag(s))}{s-\xi}\dd s\right\}.
\end{aligned}
\ee
Define $\hat{L}\doteq\{z=\kappa\e^{-\frac{\ii\pi}{4}}:0<\kappa<+\infty\}$, then we have the following lemma.
\begin{lemma}\label{lem3.2}
As $t\rightarrow\infty$, for $z\in\hat{L}$, the estimate for $\tilde{\delta}(\xi)$ holds:
\be
\left|(T_{\xi_0}\tilde{\delta})(z)\right|\leq Ct^{-\frac{1}{2}}.\label{3.40}
\ee
\end{lemma}
\begin{proof}
See Lemma 3.2 in \cite{LZG}.
\end{proof}
\begin{remark}
In a similar way, we have for $z\in\hat{L}^*$,
\be\label{3.41}
\left|(T_{\xi_0}\hat{\delta})(z)\right|\leq Ct^{-\frac{1}{2}},\quad t\rightarrow\infty,
\ee
where $\hat{\delta}(\xi)=\left((\delta^{-1}-(\det\delta)^{-1}\mathbb{I}_{3\times3})\gamma^\dag_{a}\e^{t\Phi}\right)(\xi)$. Analogous estimates also follow for $T_{-\xi_0}$.
\end{remark}
Let $\mathcal{X}_{\pm\xi_0}\doteq X\pm\xi_0$ be the cross $X$ defined by \eqref{3.2} centered at $\pm\xi_0$ and $D_\epsilon(\pm\xi_0)$ denote the open disk of radius $\epsilon$ centered at $\pm\xi_0$ for a small $\epsilon>0$, $\mathcal{X}_{\pm\xi_0}^\epsilon=\mathcal{X}_{\pm\xi_0}\cap D_\epsilon(\pm\xi_0)$. Then, combining above analysis, we find that as $t\to\infty$, the jump matrix $J^{(2)}(x,t;\xi)$ in $D_\epsilon(\xi_0)$ tends to $J^{(\xi_0)}(\gamma(\xi_0);z)$ for $z\in X$, where
\be
J^{(\xi_0)}(\gamma(\xi_0);z)=\left\{
\begin{aligned}
&\begin{pmatrix}
1 & \vartheta^{2}z^{-2\ii\nu}\e^{-\frac{\ii z^2}{2}}\gamma(\xi_0)\\
\textbf{0}_{3\times1} & \mathbb{I}_{3\times3}
\end{pmatrix},\quad \xi\in\left(\mathcal{X}_{\xi_0}^\epsilon\right)_1,\\
&\begin{pmatrix}
1 & -\vartheta^2z^{-2\ii\nu}\e^{-\frac{\ii z^2}{2}}\frac{\gamma(\xi_0)}{1+\gamma(\xi_0)\gamma^\dag(\xi_0)}\\
\textbf{0}_{3\times1} & \mathbb{I}_{3\times3}
\end{pmatrix},\quad\,\ \xi\in\left(\mathcal{X}_{\xi_0}^\epsilon\right)_2,\\
&\begin{pmatrix}
1 & \textbf{0}_{1\times3}\\
-\vartheta^{-2}z^{2\ii\nu}\e^{\frac{\ii z^2}{2}}\frac{\gamma^\dag(\xi_0)}{1+\gamma(\xi_0)\gamma^\dag(\xi_0)} & \mathbb{I}_{3\times3}
\end{pmatrix},\qquad \xi\in\left(\mathcal{X}_{\xi_0}^\epsilon\right)_3,\\
&\begin{pmatrix}
1 & \textbf{0}_{1\times3}\\
\vartheta^{-2}z^{2\ii\nu}\e^{\frac{\ii z^2}{2}}\gamma^\dag(\xi_0) & \mathbb{I}_{3\times3}
\end{pmatrix},\quad \xi\in \left(\mathcal{X}_{\xi_0}^\epsilon\right)_4.
\end{aligned}
\right.
\ee
This suggests that in the neighborhood $D_\epsilon(\xi_0)$ of $\xi_0$, we can approximate $M^{(2)}(x,t;\xi)$ by a $4\times4$ matrix-valued function of the form
\be\label{3.43}
M^{(\xi_0)}(x,t;\xi)=\vartheta^{-\sigma}M^X(\gamma(\xi_0);z)\vartheta^{\sigma},
\ee
where $M^X(\gamma(\xi_0);z)$ is the solution of the model RH problem considered in Subsection \ref{sec3.1} by setting $\rho=\gamma(\xi_0)$.
\begin{lemma}\label{lem3.3}
The function $M^{(\xi_0)}(x,t;\xi)$ defined in \eqref{3.43} is analytic and bounded for $\xi\in D_\epsilon(\xi_0)\setminus\mathcal{X}_{\xi_0}^\epsilon$. On the contour $\mathcal{X}_{\xi_0}^\epsilon$, $M^{(\xi_0)}$ satisfies the jump relation $M_+^{(\xi_0)}=M_-^{(\xi_0)}J^{(\xi_0)}$, and $J^{(\xi_0)}$ obeys the estimate for $1\leq n\leq\infty$:
\be\label{3.44}
\left\|J^{(2)}-J^{(\xi_0)}\right\|_{L^n\left(\mathcal{X}_{\xi_0}^\epsilon\right)}\leq Ct^{-\frac{1}{2}-\frac{1}{2n}}\ln t.
\ee
As $t\rightarrow\infty$, we have
\be\label{3.45}
\left\|\left[M^{(\xi_0)}(x,t;\xi)\right]^{-1}-\mathbb{I}_{4\times4}\right\|_{L^\infty(\partial D_\epsilon(\xi_0))}=O(t^{-\frac{1}{2}}),
\ee
and
\be\label{3.46}
-\frac{1}{2\pi\ii}\int_{\partial D_\epsilon(\xi_0)}\left(\left[M^{(\xi_0)}(x,t;\xi)\right]^{-1}-\mathbb{I}_{4\times4}\right)\dd \xi=\frac{\vartheta^{-\sigma}M^X_1(\gamma(\xi_0))\vartheta^{\sigma}}
{\sqrt{48t\xi_0}}+O(t^{-1}).
\ee
\end{lemma}
\begin{proof}
The analyticity and bound of $M^{(\xi_0)}$ are a consequence of Theorem \ref{th3.1}. Moreover, for $\xi\in(\mathcal{X}_{\xi_0}^\epsilon)_1$, that is, $z=\sqrt{48t\xi_0}\varsigma\e^{-\frac{\ii\pi}{4}}$, $0\leq \varsigma\leq\epsilon$, it follows from the Lemma 3.35 in \cite{PD} that
\be\label{3.47}
\left\|\upsilon^2\gamma_a\left(\frac{z}{\sqrt{48t\xi_0}}+\xi_0\right)-z^{-2\ii\nu}\e^{-\frac{\ii z^2}{2}}\gamma(\xi_0)\right\|_{L^\infty\left((\mathcal{X}_{\xi_0}^\epsilon)_1\right)}\leq C\left|\e^{-\frac{\ii\kappa}{2}z^2}\right|t^{-\frac{1}{2}}\ln t,\quad 0<\kappa<\frac{1}{2}.
\ee
Thus, together with \eqref{3.40} and \eqref{3.41}, we have
\be\label{3.48}
\left\|J^{(2)}-J^{(\xi_0)}\right\|_{L^\infty\left((\mathcal{X}_{\xi_0}^\epsilon)_1\right)}\leq C\left|\e^{-\frac{\ii\kappa}{2}z^2}\right|t^{-\frac{1}{2}}\ln t.
\ee
Hence, we arrive at
\be\label{3.49}
\left\|J^{(2)}-J^{(\xi_0)}\right\|_{L^1\left((\mathcal{X}_{\xi_0}^\epsilon)_1\right)}\leq Ct^{-1}\ln t.
\ee
By the general inequality $\|f\|_{L^n}\leq\|f\|^{1-1/n}_{L^\infty}\|f\|_{L^1}^{1/n}$, we find
\be\label{3.50}
\left\|J^{(2)}-J^{(\xi_0)}\right\|_{L^n\left((\mathcal{X}_{\xi_0}^\epsilon)_1\right)}\leq Ct^{-\frac{1}{2}-\frac{1}{2n}}\ln t.
\ee
The norms on $(\mathcal{X}_{\xi_0}^\epsilon)_j$, $j=2,3,4$ can be estimated in a similar way. Therefore, \eqref{3.44} follows.

If $\xi\in\partial D_\epsilon(\xi_0)$, the variable $z=\sqrt{48t\xi_0}(\xi-\xi_0)$ tends to infinity as $t\rightarrow\infty$. It follows from \eqref{3.4} that
\be
M^X(\gamma(\xi_0);z)=\mathbb{I}_{4\times4}+\frac{M^X_1(\gamma(\xi_0))}{\sqrt{48t\xi_0}(\xi-\xi_0)}+O(t^{-1}),\quad t\rightarrow\infty,~\xi\in \partial D_\epsilon(\xi_0).
\ee
Then the equation \eqref{3.43} yields
\be\label{3.52}
\left[M^{(\xi_0)}(x,t;\xi)\right]^{-1}-\mathbb{I}_{4\times4}=-\frac{\vartheta^{-\sigma}M^X_1(\gamma(\xi_0))\vartheta^{\sigma}}
{\sqrt{48t\xi_0}(\xi-\xi_0)}+O(t^{-1}),\quad t\rightarrow\infty,~\xi\in \partial D_\epsilon(\xi_0).
\ee
The estimate \eqref{3.45} then immediately follows. By Cauchy's formula and \eqref{3.52}, we find \eqref{3.46}.
\end{proof}
On the other hand, it is easy to check that as $t\to\infty$, the jump matrix $J^{(2)}(x,t;\xi)$ in $D_\epsilon(-\xi_0)$ tends to $\mathcal{A}\left[J^{(\xi_0)}(\gamma(\xi_0);-z^*)\right]^*\mathcal{A}$ for $z\in X$. Thus, we can approximate $M^{(2)}(x,t;\xi)$ in the neighborhood $D_\epsilon(-\xi_0)$ of $-\xi_0$ by $\mathcal{A}\left[M^{(\xi_0)}(x,t;-\xi^*)\right]^*\mathcal{A}$.

\subsection{Find asymptotic formula}\label{sec3.4}
Define the approximate solution $M^{(app)}(x,t;\xi)$ by
\be\label{3.53}
M^{(app)}(x,t;\xi)=\left\{\begin{aligned}
&M^{(\xi_0)}(x,t;\xi),\qquad\qquad\quad \xi\in D_\epsilon(\xi_0),\\
&\mathcal{A}\left[M^{(\xi_0)}(x,t;-\xi^*)\right]^*\mathcal{A},\ \xi\in D_\epsilon(-\xi_0),\\
&\mathbb{I}_{4\times4},\qquad \qquad\qquad\,\,\ \qquad{\text elsewhere}.
\end{aligned}
\right.
\ee
Let $\hat{M}(x,t;\xi)$ be
\be\label{3.54}
\hat{M}(x,t;\xi)=M^{(2)}(x,t;\xi)\left[M^{(app)}(x,t;\xi)\right]^{-1}.
\ee
Denote $\hat{\Sigma}=\Sigma\cup\partial D_\epsilon(-\xi_0)\cup\partial D_\epsilon(\xi_0)$ and suppose that the boundaries of $D_\epsilon(\pm\xi_0)$ are oriented counterclockwise, see Figure \ref{fig4}. Then $\hat{M}(x,t;\xi)$ satisfies the following $4\times4$ matrix RH problem:\\
$\bullet$ $\hat{M}(x,t;\xi)$ is analytic for $\xi\in\bfC\setminus\hat{\Sigma}$ with continuous boundary values on $\hat{\Sigma}$;\\
$\bullet$ Across the contour $\hat{\Sigma}$, the limiting values $\hat{M}_\pm(x,t;\xi)$ obey the jump relation
\be
\hat{M}_+(x,t;\xi)=\hat{M}_-(x,t;\xi)\hat{J}(x,t;\xi);
\ee
$\bullet$ As $\xi\to\infty$, $\hat{M}(x,t;\xi)\to\mathbb{I}_{4\times4}$;\\
where the jump matrix $\hat{J}(x,t;\xi)$ is given by
\be\label{3.55}
\hat{J}=\left\{
\begin{aligned}
&M^{(app)}_-J^{(2)}\left[M^{(app)}_+\right]^{-1},\,\,\ \xi\in\hat{\Sigma}\cap (D_\epsilon(\xi_0)\cup D_\epsilon(-\xi_0)),\\
&\left[M^{(app)}\right]^{-1},\qquad\qquad\quad \xi\in\partial D_\epsilon(\xi_0)\cup\partial D_\epsilon(-\xi_0),\\
&J^{(2)},\qquad\qquad\qquad\qquad\quad \xi\in\hat{\Sigma}\setminus (\overline{D_\epsilon(\xi_0)}\cup\overline{D_\epsilon(-\xi_0)}).
\end{aligned}
\right.
\ee
\begin{figure}[htbp]
  \centering
 \includegraphics[width=3.5in]{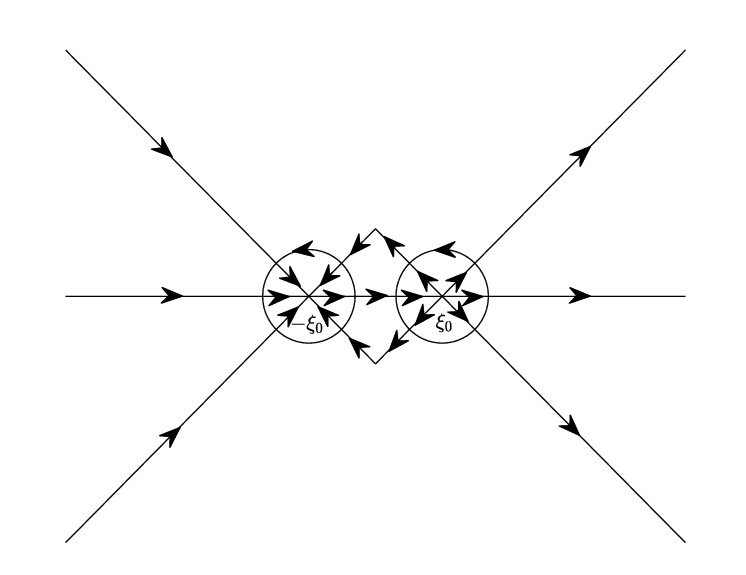}
   \caption{The oriented contour $\hat{\Sigma}$.}\label{fig4}
\end{figure}

Now we rewrite $\hat{\Sigma}$ as:
$$\hat{\Sigma}=(\partial D_\epsilon(-\xi_0)\cup\partial D_\epsilon(\xi_0))\cup(\mathcal{X}_{-\xi_0}^\epsilon\cup\mathcal{X}_{\xi_0}^\epsilon)
\cup\tilde{\Sigma}\cup\bfR,$$
where
$$\tilde{\Sigma}=\bigcup_1^4\Sigma_j\setminus(D_\epsilon(-\xi_0)\cup D_\epsilon(\xi_0)),$$
and denote $\hat{w}=\hat{J}-\mathbb{I}_{4\times4}$, then we have the following lemma.
\begin{lemma}\label{lem3.4}
For $1\leq n\leq\infty$, the following estimates hold:
\begin{align}
\|\hat{w}\|_{L^n(\partial D_\epsilon(-\xi_0)\cup\partial D_\epsilon(\xi_0))}&\leq Ct^{-\frac{1}{2}},\label{3.57}\\
\|\hat{w}\|_{L^n(\mathcal{X}_{-\xi_0}^\epsilon\cup\mathcal{X}_{\xi_0}^\epsilon)}&\leq Ct^{-\frac{1}{2}-\frac{1}{2n}}\ln t,\label{3.58}\\
\|\hat{w}\|_{L^n(\tilde{\Sigma})}&\leq C\e^{-ct},\label{3.59}\\
\|\hat{w}\|_{L^n(\bfR)}&\leq Ct^{-\frac{3}{2}}.\label{3.60}
\end{align}
\end{lemma}
\begin{proof}
It follows the same line as Lemma 3.4 in \cite{LZG}.
\end{proof}

The uniformly vanishing bound on $\hat{w}$ in Lemma \ref{lem3.4} establishes RH problem for $\hat{M}(x,t;\xi)$ as a small-norm Riemann--Hilbert problem, for which there is a well known existence and uniqueness theorem \cite{PD}. In fact, we may write
\be\label{3.61}
\hat{M}(x,t;\xi)=\mathbb{I}_{4\times4}+\frac{1}{2\pi\ii}\int_{\hat{\Sigma}}
\frac{(\hat{\mu}\hat{w})(x,t;s)}{s-\xi}\dd s,
\ee
where the $4\times4$ matrix-valued function $\hat{\mu}(x,t;\xi)$ is the unique solution of
\be\label{3.62}
\hat{\mu}=\mathbb{I}_{4\times4}+\hat{C}_{\hat{w}}\hat{\mu}.
\ee
The singular integral operator $\hat{C}_{\hat{w}}$: $L^2(\hat{\Sigma})\to L^2(\hat{\Sigma})$ is defined for $f\in L^2(\hat{\Sigma})$ by
\begin{align}
\hat{C}_{\hat{w}}f&\doteq\hat{C}_-(f\hat{w}),\label{3.63}\\
(\hat{C}_-f)(\xi)&\doteq\lim_{\xi\to\hat{\Sigma}_-}\int_{\hat{\Sigma}}\frac{f(s)}{s-\xi}\frac{\dd s}{2\pi\ii},
\end{align}
where $\hat{C}_-$ is the well known Cauchy operator. Then, by Lemma \ref{lem3.4} and \eqref{3.63}, we find
\be\label{3.65}
\|\hat{C}_{\hat{w}}\|_{B(L^2(\hat{\Sigma}))}\leq C\|\hat{w}\|_{L^\infty(\hat{\Sigma})}\leq Ct^{-\frac{1}{2}}\ln t,
\ee
where $B(L^2(\hat{\Sigma}))$ denotes the Banach space of bounded linear operators $L^2(\hat{\Sigma})\rightarrow L^2(\hat{\Sigma})$. Therefore, there exists a $T>0$ such that $1-\hat{C}_{\hat{w}}\in B(L^2(\hat{\Sigma}))$ is invertible for all $(x,t)\in\mathcal{O},$ $t>T$. And hence the existence of both $\hat{\mu}$ and $\hat{M}$ immediately follow.

Moreover, standard estimates using the Neumann series shows that $\hat{\mu}(x,t;\xi)$ satisfies
\be\label{3.66}
\|\hat{\mu}(x,t;\cdot)-\mathbb{I}_{4\times4}\|_{L^2(\hat{\Sigma})}=O(t^{-\frac{1}{2}}),\quad t\rightarrow\infty.
\ee
In fact, Equation \eqref{3.62} is equivalent to $\hat{\mu}=\mathbb{I}_{4\times4}+(1-\hat{C}_{\hat{w}})^{-1}\hat{C}_{\hat{w}}\mathbb{I}_{4\times4}$. Using the Neumann series, one can obtain
$$\|(1-\hat{C}_{\hat{w}})^{-1}\|_{B(L^2(\hat{\Sigma}))}\leq\frac{1}{1-\|\hat{C}_{\hat{w}}\|_{B(L^2(\hat{\Sigma}))}}$$
whenever $\|\hat{C}_{\hat{w}}\|_{B(L^2(\hat{\Sigma}))}<1$. Then, we find
\bea
\begin{aligned}
\|\hat{\mu}(x,t;\cdot)-\mathbb{I}_{4\times4}\|_{L^2(\hat{\Sigma})}
=&\|(1-\hat{C}_{\hat{w}})^{-1}\hat{C}_{\hat{w}}\mathbb{I}_{4\times4}\|_{L^2(\hat{\Sigma})}\\
\leq&\|(1-\hat{C}_{\hat{w}})^{-1}\|_{B(L^2(\hat{\Sigma}))}\|\hat{C}_-(\hat{w})\|_{L^2(\hat{\Sigma})}\\
\leq&\frac{C\|\hat{w}\|_{L^2(\hat{\Sigma})}}{1-\|\hat{C}_{\hat{w}}\|_{B(L^2(\hat{\Sigma}))}}\leq C\|\hat{w}\|_{L^2(\hat{\Sigma})}
\end{aligned}
\eea
for all $t$ large enough. In view of Lemma \ref{lem3.4}, this gives \eqref{3.66}.

It then follows from \eqref{3.10}, \eqref{3.23}, \eqref{3.54}, \eqref{3.53} and \eqref{3.61} that
\be\label{3.68}
\lim_{\xi\rightarrow\infty}\xi(M(x,t;\xi)-\mathbb{I}_{4\times4})=
\lim_{\xi\rightarrow\infty}\xi(\hat{M}(x,t;\xi)-\mathbb{I}_{4\times4})
=-\frac{1}{2\pi\ii}\int_{\hat{\Sigma}}(\hat{\mu}\hat{w})(x,t;\xi)\dd\xi.
\ee
By \eqref{3.46}, \eqref{3.53}, \eqref{3.57} and \eqref{3.66}, we can get
\be
\begin{aligned}
&-\frac{1}{2\pi\ii}\int_{\partial D_\epsilon(-\xi_0)\cup\partial D_\epsilon(\xi_0)}(\hat{\mu}\hat{w})(x,t;\xi)\dd\xi\\
=&-\frac{1}{2\pi\ii}\int_{\partial D_\epsilon(-\xi_0)\cup\partial D_\epsilon(\xi_0)}\hat{w}\dd \xi
-\frac{1}{2\pi\ii}\int_{\partial D_\epsilon(-\xi_0)\cup\partial D_\epsilon(\xi_0)}(\hat{\mu}-\mathbb{I}_{4\times4})\hat{w}\dd\xi\\
=&-\frac{1}{2\pi\ii}\int_{\partial D_\epsilon(-\xi_0)}\left(\left(\mathcal{A}\left[M^{(\xi_0)}(x,t;-\xi^*)\right]^*\mathcal{A}\right)^{-1}-\mathbb{I}_{4\times4}\right)\dd \xi\\
&-\frac{1}{2\pi\ii}\int_{\partial D_\epsilon(\xi_0)}\left(\left(M^{(\xi_0)}(x,t;\xi)\right)^{-1}-\mathbb{I}_{4\times4}\right)\dd\xi \\
&+O(\|\hat{\mu}-\mathbb{I}_{4\times4}\|_{L^2(\partial D_\epsilon(-\xi_0)\cup\partial D_\epsilon(\xi_0))}
\|\hat{w}\|_{L^2(\partial D_\epsilon(-k_0)\cup\partial D_\epsilon(\xi_0))})\\
=&-\frac{\mathcal{A}\left[\vartheta^{-\sigma}M^X_1(\gamma(\xi_0))\vartheta^{\sigma}\right]^*\mathcal{A}}
{\sqrt{48t\xi_0}}+\frac{\vartheta^{-\sigma}M^X_1(\gamma(\xi_0))\vartheta^{\sigma}}
{\sqrt{48t\xi_0}}+O(t^{-1}).
\end{aligned}
\ee
Similarly, the contributions from $\mathcal{X}_{-k_0}^\epsilon\cup\mathcal{X}_{k_0}^\epsilon$, $\tilde{\Sigma}$ and $\bfR$ to the right-hand side of \eqref{3.68} are $O(t^{-1}\ln t)$, $O(\e^{-ct})$ and $O(t^{-3/2})$, respectively.

Thus, taking into account that the reconstructional formula \eqref{2.22} and \eqref{3.5}, we get
\begin{align} \label{3.71}
\begin{pmatrix}
u(x,t) & u^*(x,t) & w(x,t)
\end{pmatrix}=&2\ii\lim_{\xi\rightarrow\infty}(\xi M(x,t;\xi))_{12}\nn\\
=&2\ii\lim_{\xi\rightarrow\infty}\xi(\hat{M}(x,t;\xi)-\mathbb{I}_{4\times4})_{12}\\
=&2\ii\left(\frac{-\ii\vartheta^2\beta^X-\ii\left(\vartheta^2\beta^X\right)^*\sigma_1}
{\sqrt{48t\xi_0}}\right)+O\left(\frac{\ln t}{t}\right).\nn
\end{align}
Collecting above computations, we obtain the long-time asymptotic result of the solution in oscillating sector $\mathcal{O}$.
\begin{theorem}\label{th3.2}
Let $u_0(x)$ and $w_0(x)$ lie in the Schwartz space $\mathcal{S}(\bfR)$, and generate the scattering data in sense that: the determinant of the $3\times3$ matrix-valued spectral function $a(\xi)$ defined in \eqref{2.17} has no zeros in $\bfC_-$. Then, in oscillating sector $\mathcal{O}$, as $t\rightarrow\infty$, the solution of the initial problem for the new two-component Sasa-Satsuma equation \eqref{SS} on the line satisfies the following asymptotic formula
\be\label{3.71}
\begin{pmatrix}u(x,t) & w(x,t)\end{pmatrix}=\frac{\begin{pmatrix}u_{as}(x,t) & w_{as}(x,t)\end{pmatrix}}{\sqrt{t}}+O\left(\frac{\ln t}{t}\right),
\ee
where the leading-order coefficient $u_{as}(x,t)$ and $w_{as}(x,t)$ are given by
\begin{align}
u_{as}(x,t)=&\frac{\nu\e^{-\frac{\pi\nu}{2}}}{\sqrt{24\pi \xi_0}}\left((192t\xi_0^3)^{\ii\nu}\e^{16\ii t\xi_0^3+2\chi(\xi_0)+\frac{\pi\ii}{4}}\Gamma(-\ii\nu)\gamma_1(\xi_0)\right.\\
&\left.+(192t\xi_0^3)^{-\ii\nu}\e^{-16\ii t\xi_0^3-2\chi(\xi_0)-\frac{\pi\ii}{4}}\Gamma(\ii\nu)\gamma_1(-\xi_0)\right),\nn\\
w_{as}(x,t)=&\frac{\nu\e^{-\frac{\pi\nu}{2}}}{\sqrt{6\pi \xi_0}}\left|\Gamma(-\ii\nu)\gamma_2(\xi_0)\right|
\cos\bigg(\nu\ln(192t\xi_0^3)+16t\xi_0^3+\frac{\pi}{4}+\arg \gamma_2(\xi_0)\\
&+\arg\Gamma(-\ii\nu)-
\frac{1}{\pi}\int_{-\xi_0}^{\xi_0}\ln\left(\frac{1+\gamma(s)\gamma^\dag(s)}
{1+\gamma(\xi_0)\gamma^\dag(\xi_0)}\right)\frac{\dd s}{s-\xi_0}\bigg),\nn
\end{align}
where $\xi_0$, $\nu$, $\chi(\xi)$ and $\gamma_1(\xi)$, $\gamma_2(\xi)$ are given by \eqref{3.8}, \eqref{3.31}, \eqref{3.32} and \eqref{3.17}, respectively.
\end{theorem}

\section{Asymptotic analysis in Painlev\'e sector $\mathcal{P}$}\label{sec4}
In this section, we study the long-time asymptotics of solution to the new two-component Sasa--Satsuma equation \eqref{SS} in Painlev\'e region $\mathcal{P}$ defined by
\be\label{4.1}
\mathcal{P}=\{(x,t)\in\bfR^2|t>1,\ |x|\leq Nt^{1/3}\},\ N \ \text{constant}.
\ee
Let
\berr
\mathcal P_{\geq}\doteq\mathcal P\cap\{x\geq0\},\quad \mathcal P_{\leq}\doteq\mathcal P\cap\{x\leq0\}
\eerr
denote the right and left halves of $\mathcal P$. The long-time asymptotic formula of the solution for the case $(x,t)\in\mathcal P_{\geq}$ will be established, the case when $(x,t)\in\mathcal P_{\leq}$ can be handled in a similar but easy way.

At first, we will prove Theorem \ref{theorem4.1}, which expresses the large $z$ behavior of the solution of a model RH problem in terms of the solution of a new coupled Painlev\'e II equation. This result will be important for analyzing the long-time asymptotics of the solution of system \eqref{SS} in region $\mathcal{P}_\geq$.
\subsection{Another model RH problem}\label{sec4.1}
\begin{figure}[htbp]
\centering
\includegraphics[width=3.5in]{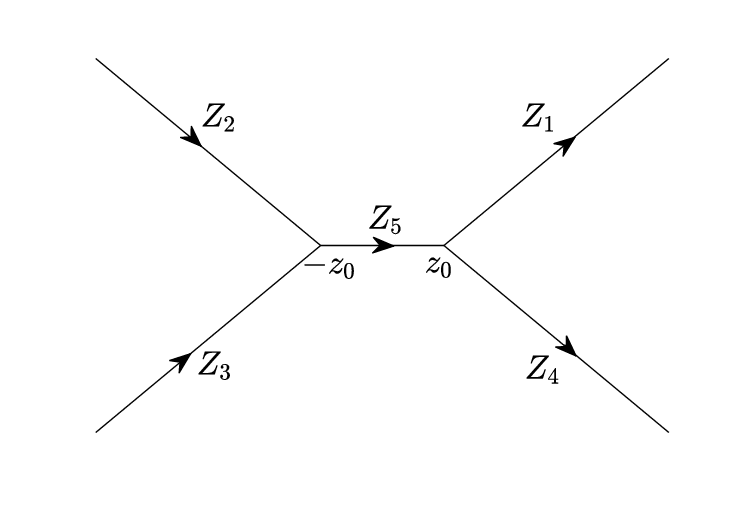}
\caption{The oriented contour $Z$.}\label{fig5}
\end{figure}
Given $z_0\geq0$, let $Z$ denote the contour $Z=\cup_{j=1}^5Z_j$, where the line segments
\be\label{4.2}
\begin{aligned}
Z_1&=\{z_0+\varsigma\e^{\frac{\pi\ii}{6}}|\varsigma\geq0\},\qquad\,\,\
Z_2=\{-z_0+\varsigma\e^{\frac{5\pi\ii}{6}}|\varsigma\geq0\},\\
Z_3&=\{-z_0+\varsigma\e^{-\frac{5\pi\ii}{6}}|\varsigma\geq0\},\quad
Z_4=\{z_0+\varsigma\e^{-\frac{\pi\ii}{6}}|\varsigma\geq0\},\\
Z_5&=\{\varsigma|-z_0\leq \varsigma\leq z_0\},
\end{aligned}
\ee
are oriented as in Figure \ref{fig5}. We consider the following family of RH problems parameterized by $y>0$ and a $1\times3$ complex-valued row vector $p$ with $p=p^*\sigma_1$ for $M^Z$:\\
$\bullet$ $M^Z(y,p;z)$ is analytic in $\bfC\setminus Z$ with continuous boundary values on $Z$;\\
$\bullet$ For $z\in Z$, the boundary values $M^Z_\pm$ satisfy the jump relation
$M^Z_+(y,p;z)=M^Z_-(y,p;z)J^Z(y,p;z);$\\
$\bullet$  $M^Z(y,p;z)\rightarrow \mathbb{I}_{4\times4}$, as $z\rightarrow\infty$;\\
where the jump matrix $J^Z$ is defined by
\be\label{4.3}
J^Z(y,p;z)=\left\{
\begin{aligned}
&\begin{pmatrix}
1 & \textbf{0}_{1\times3}\\
p^\dag\e^{2\ii(\frac{4}{3}z^3-yz)} & \mathbb{I}_{3\times3}
\end{pmatrix},\quad z\in Z_1\cup Z_2,\\
&\begin{pmatrix}
1 & p\e^{-2\ii(\frac{4}{3}z^3-yz)}\\
\textbf{0}_{3\times1} & \mathbb{I}_{3\times3}
\end{pmatrix},\quad z\in Z_3\cup Z_4,\\
&\begin{pmatrix}
1 & p\e^{-2\ii(\frac{4}{3}z^3-yz)}\\
\textbf{0}_{3\times1} & \mathbb{I}_{3\times3}
\end{pmatrix}\begin{pmatrix}
1 & \textbf{0}_{1\times3}\\
p^\dag\e^{2\ii(\frac{4}{3}z^3-yz)} & \mathbb{I}_{3\times3}
\end{pmatrix},\quad z\in Z_5.
\end{aligned}
\right.
\ee
\begin{theorem}\label{theorem4.1}
Define the parameter subset
\be\label{4.4}
\Bbb P=\{(y,t,z_0)\in\bfR^3|0\leq y\leq C_1,\ t\geq2,\ \sqrt{y}/2\leq z_0\leq C_2\},
\ee
where $C_1$, $C_2>0$ are constants. Then the RH problem for $M^Z$ with jump matrix $J^Z$ has a unique solution $M^Z(y,p;z)$ whenever $(y,t,z_0)\in\Bbb P$. There are smooth functions $\{M_j^Z(y)\}$ such that
\be\label{4.5}
M^Z(y,p;z)=\mathbb{I}_{4\times4}+\sum_{j=1}^N\frac{M_j^Z(y)}{z^j}+O(z^{-N-1}), \quad z\rightarrow\infty,
\ee
however, the $(1,2)$ entry and $(1,4)$ entry of leading coefficient $M_1^Z$ are given by
\be\label{4.6}
\left(M_1^Z(y)\right)_{12}=u_p(y),\quad \left(M_1^Z(y)\right)_{14}= \ii w_p(y),
\ee
where $u_p(y)$ and $w_p(y)$ are complex-valued and real-valued functions, respectively, and are the smooth solution of a new coupled Painlev\'e II equation \eqref{B.5}. Moreover, $M^Z(y,p;z)$ is uniformly bounded for $z\in\bfC\setminus Z$, and satisfies the symmetries
\be\label{4.7}
M^Z(y,p;z)=\left[\left(M^Z\right)^\dag(y,p;z^*)\right]^{-1}
=\mathcal{A}\left[M^Z(y,p;-z^*)\right]^*\mathcal{A}.
\ee
\end{theorem}
\begin{proof}
It is easy to see that
\begin{align}
\text{Re}\left(2\ii\left(\frac{4}{3}z^3-yz\right)\right)
=\varsigma\left(-\frac{8}{3}\varsigma^2-4\sqrt{3}z_0\varsigma-4z_0^2+y\right)
\leq-\frac{8}{3}\varsigma^3-4\sqrt{3}z_0\varsigma^2,\nn
\end{align}
for all $z=z_0+\varsigma\e^{\frac{\pi\ii}{6}}$ and $z=-z_0+\varsigma\e^{\frac{5\pi\ii}{6}}$ with $\varsigma\geq0,z_0\geq0$ and $0<y\leq4z_0^2$. Thus, we have
\berr
|\e^{2\ii(\frac{4}{3}z^3+yz)}|\leq C\e^{-|z\pm z_0|^2(|z\pm z_0|+z_0)},\quad z\in Z_1\cup Z_2.
\eerr
Analogous estimates hold for $z\in Z_3\cup Z_4$.
This shows that $J^Z\to \mathbb{I}_{4\times4}$ exponentially fast as $z\rightarrow\infty$.

Note that, the jump matrix $J^Z$ satisfies the same symmetries \eqref{B.6} and \eqref{B.14} as $J^P$. In other words, $J^Z$ is Hermitian and positive definite on $Z\cap\bfR$ and satisfies $J^Z(y,p;z)=(J^Z)^\dag(y,p;z^*)$ on $Z\setminus\bfR$. This implies that the jump conditions and the jump matrices in the RH problem for $M^Z$ satisfy the hypotheses of Zhou's vanishing lemma \cite{ZX}, that is, the jump contour $Z_1\cup Z_2$ has the necessary invariance under Schwarz reflection with orientation and if $z$ lies in the part of jump contour on the real axis, $J^Z+(J^Z)^\dag$ is positive definite. Thus we deduce the unique existence of the solution $M^Z$. The symmetries of $J^Z$ implies that \eqref{4.7} follows. Moreover, the RH problem for $M^Z(y,p;z)$ can be transformed into the RH problem for $M^P(y;z)$ stated in Appendix \ref{secB} up to a trivial contour deformation. Therefore, we complete the proof of Theorem \ref{theorem4.1}.
\end{proof}

\subsection{Transformations}

Suppose $(x,t)\in\mathcal P_\geq$. Then, as $t\rightarrow\infty$, the critical points $\pm \xi_0$ given by \eqref{3.8} approach 0. In this case, we only need the triangular factorization of the jump matrix in the form as follows:
\be
J(x,t;\xi)=\begin{pmatrix}
1 & \gamma(\xi)\e^{-t\Phi(x,t;\xi)}\\
\textbf{0}_{3\times1} & \mathbb{I}_{3\times3}
\end{pmatrix}\begin{pmatrix}
1 & \textbf{0}_{1\times3} \\
\gamma^\dag(\xi)\e^{t\Phi(x,t;\xi)} & \mathbb{I}_{3\times3}
\end{pmatrix}.
\ee
\begin{figure}[htbp]
\centering
\includegraphics[width=3.5in]{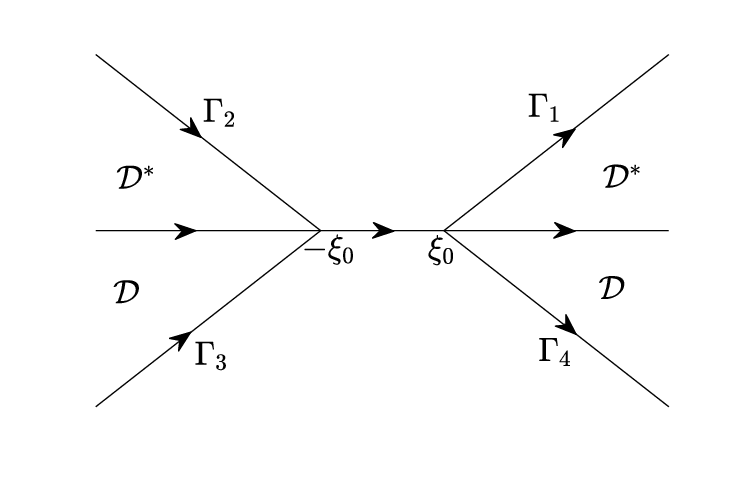}
\caption{The sets $\mathcal D$ and $\mathcal D^*$ and oriented contour $\Gamma$.}\label{fig6}
\end{figure}
Define the contour $\Gamma\subset\bfC$ by $\Gamma=\cup_{j=1}^4\Gamma_j\cup\bfR$, where
\be
\begin{aligned}
\Gamma_1&=\{\xi_0+\varsigma\e^{\frac{\pi\ii}{6}}|\varsigma\geq0\},\quad\quad\,\,\
\Gamma_2=\{-\xi_0+\varsigma\e^{\frac{5\pi\ii}{6}}|\varsigma\geq0\},\\
\Gamma_3&=\{-\xi_0+\varsigma\e^{-\frac{5\pi\ii}{6}}|\varsigma\geq0\},\quad
\Gamma_4=\{\xi_0+\varsigma\e^{-\frac{\pi\ii}{6}}|\varsigma\geq0\}.
\end{aligned}
\ee
The orientation of $\Gamma$ and the triangular domains $\mathcal{D}$, $\mathcal{D}^*$ are shown in Figure \ref{fig6}. Recalling \eqref{3.17}, by Lemma \ref{lem3.1}, we also have the following analytic decomposition lemma for $\gamma_1(\xi)$ and $\gamma_2(\xi)$.
\begin{lemma}
For $j=1,2$, we have
\be
\gamma_j(\xi)=\gamma_{j,a}(x,t;\xi)+\gamma_{j,r}(x,t;\xi), \quad \xi\in(-\infty,-\xi_0)\cup(\xi_0,\infty),
\ee
where the functions $\gamma_{j,a}$ and $\gamma_{j,r}$ satisfy:\\
(i) For $(x,t)\in\mathcal{P}_\geq$, $\gamma_{j,a}(x,t;\xi)$ is defined and continuous for $\xi\in\bar{\mathcal D}$ and analytic for $\xi\in\mathcal D$.\\
(ii) The function $\gamma_{j,a}$ satisfies
\be
|\gamma_{j,a}(x,t;\xi)|\leq \frac{C}{1+|\xi|^2}\e^{\frac{t}{4}|\text{Re}\Phi(x,t;\xi)|},~\xi\in\bar{\mathcal D},
\ee
and
\be
|\gamma_{j,a}(x,t;\xi)-\gamma(\xi_0)|\leq C|\xi-\xi_0|\e^{\frac{t}{4}|\text{Re}\Phi(x,t;\xi)|},~\xi\in\bar{\mathcal D}.
\ee
(iii) The $L^1, L^2$ and $L^\infty$ norms of the function $\gamma_{j,r}(x,t;\cdot)$ on $(-\infty,-\xi_0)\cup(\xi_0,\infty)$ are $O(t^{-3/2})$ as $t\rightarrow\infty$ uniformly with respect to $(x,t)\in\mathcal{P}_\geq$.
\end{lemma}
Thus, we have obtained a decomposition $\gamma(\xi)=\gamma_{a}(x,t;\xi)+\gamma_{r}(x,t;\xi)$ by setting
\begin{align*}
\gamma_{a}(x,t;\xi)=&(\gamma_{1,a}(x,t;\xi),\gamma^*_{1,a}(x,t;-\xi^*),\gamma_{2,a}(x,t;\xi)),\quad \xi\in\mathcal D,\\
\gamma_{r}(x,t;\xi)=&(\gamma_{1,r}(x,t;\xi),\gamma^*_{1,r}(x,t;-\xi),\gamma_{2,r}(x,t;\xi)),\quad \xi\in\bfR.
\end{align*}

Now we can deform the contour by introducing the new $4\times4$ matrix-valued function $M^{(1)}(x,t;\xi)$ as follows:
\be\label{4.13}
M^{(1)}(x,t;\xi)=M(x,t;\xi)\times\left\{
\begin{aligned}
&\begin{pmatrix}
1 & \gamma_a(x,t;\xi)\e^{-t\Phi(x,t;\xi)}\\
\textbf{0}_{3\times1} & \mathbb{I}_{3\times3}
\end{pmatrix},\quad\,\,\ \xi\in \mathcal D,\\
&\begin{pmatrix}
1 & \textbf{0}_{1\times3}\\
-\gamma^\dag_a(x,t;\xi^*)\e^{t\Phi(x,t;\xi)} & \mathbb{I}_{3\times3}
\end{pmatrix},\quad \xi\in \mathcal D^*,\\
&\mathbb{I}_{4\times4},\qquad\qquad\qquad\qquad\qquad\qquad\ \ \ \text{elsewhere}.
\end{aligned}
\right.
\ee
Then $M^{(1)}(x,t;\xi)$ satisfies the following RH problem:\\
$\bullet$ $M^{(1)}(x,t;\xi)$ is analytic in $\bfC\setminus\Gamma$ with continuous boundary values on $\Gamma$;\\
$\bullet$ For $\xi\in\Gamma$, the limiting values $M^{(1)}_\pm$ obey the jump condition $M^{(1)}_+(x,t;\xi)=M^{(1)}_-(x,t;\xi)J^{(1)}(x,t;\xi),$
where
\be
J^{(1)}=\left\{
\begin{aligned}
&\begin{pmatrix}
1 & \textbf{0}_{1\times3}\\
\gamma^\dag_a\e^{t\Phi} & \mathbb{I}_{3\times3}
\end{pmatrix},\ \xi\in\Gamma_1\cup\Gamma_2, \,\
\begin{pmatrix}
1 & \gamma_a\e^{-t\Phi}\\
\textbf{0}_{3\times1} & \mathbb{I}_{3\times3}
\end{pmatrix},\ \xi\in\Gamma_3\cup\Gamma_4, \\
&\begin{pmatrix}
1 & \gamma_r\e^{-t\Phi}\\
\textbf{0}_{3\times1} & \mathbb{I}_{3\times3}
\end{pmatrix}\begin{pmatrix}
1 & \textbf{0}_{1\times3}\\
\gamma^\dag_r\e^{t\Phi} & \mathbb{I}_{3\times3}
\end{pmatrix},\ \xi\in(-\infty,-\xi_0)\cup(\xi_0,\infty),\\
&\begin{pmatrix}
1 & \gamma\e^{-t\Phi}\\
\textbf{0}_{3\times1} & \mathbb{I}_{3\times3}
\end{pmatrix}\begin{pmatrix}
1 & \textbf{0}_{1\times3} \\
\gamma^\dag\e^{t\Phi} & \mathbb{I}_{3\times3}
\end{pmatrix},\,\,\ \xi\in(-\xi_0,\xi_0);
\end{aligned}
\right.
\ee
$\bullet$ As $\xi\to\infty$, $M^{(1)}(x,t;\xi)\to\mathbb{I}_{4\times4}$.

\subsection{Local model}
\begin{figure}[htbp]
\centering
\includegraphics[width=3.5in]{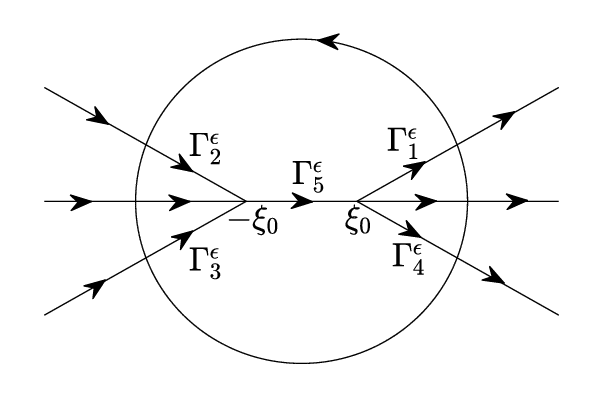}
\caption{The oriented contour $\hat{\Gamma}$ and $\Gamma^\epsilon=\cup_{j=1}^5\Gamma^\epsilon_j$.}\label{fig7}
\end{figure}
Select suitable $\epsilon>0$ and denote $D_\epsilon(0)=\{\xi\in\bfC||\xi|<\epsilon\}$. Define contour
$\Gamma^\epsilon=(\Gamma\cap D_\epsilon(0))\setminus((-\infty,-\xi_0)\cup(\xi_0,\infty))$, see Figure \ref{fig7}.
Let
\be
z\doteq(3t)^{\frac{1}{3}}\xi,\quad y\doteq(3t)^{-\frac{1}{3}}x,
\ee
such that
\be\label{4.16}
t\Phi(x,t;\xi)=2\ii\left(\frac{4}{3}z^3-yz\right).
\ee
Then for fixed $z$ and as $t\to\infty$, the jump matrix $J^{(1)}(x,t;\xi)$ tends to
\be
J^{(1)}(x,t;\xi)\rightarrow\left\{
\begin{aligned}
&\begin{pmatrix}
1 & \textbf{0}_{1\times3} \\
\gamma^\dag(0)\e^{2\ii(\frac{4}{3}z^3-yz)} & \mathbb{I}_{3\times3}
\end{pmatrix}, \quad\quad \xi\in \Gamma^\epsilon_1\cup\Gamma^\epsilon_2,\\
&\begin{pmatrix}
1 & \gamma(0)\e^{-2\ii(\frac{4}{3}z^3-yz)}\\
\textbf{0}_{3\times1} & \mathbb{I}_{3\times3}\\
\end{pmatrix},\qquad \xi\in \Gamma^\epsilon_3\cup\Gamma^\epsilon_4,\\
&\begin{pmatrix}
1 & \gamma(0)\e^{-2\ii(\frac{4}{3}z^3-yz)}\\
\textbf{0}_{3\times1} & \mathbb{I}_{3\times3}\\
\end{pmatrix}\begin{pmatrix}
1 & \textbf{0}_{1\times3} \\
\gamma^\dag(0)\e^{2\ii(\frac{4}{3}z^3-yz)} & \mathbb{I}_{3\times3}
\end{pmatrix},\,\,\xi\in \Gamma^\epsilon_5,
\end{aligned}
\right.
\ee
which is nothing but the jump matrix $J^Z$ defined in \eqref{4.3} with $p=\gamma(0)$. Thus, as $t\rightarrow\infty$, in $D_\epsilon(0)$, the solution $M^{(1)}(x,t;\xi)$ can be approximated by a $4\times4$ matrix-valued function $M^{(0)}(x,t;\xi)$ defined by
\be\label{4.18}
M^{(0)}(x,t;\xi)\doteq M^Z(y,\gamma(0);z),
\ee
where $M^Z(y,\gamma(0);z)$ is the solution of the model RH problem established in Subsection \ref{sec4.1} with $z_0=\sqrt{y}/2$. Moreover, if $(x,t)\in\mathcal{P}_\geq$, then $(y,t,z_0)\in\Bbb P$, where $\Bbb P$ is the parameter subset defined in \eqref{4.4}. Thus, Theorem \ref{theorem4.1} implies that $M^{(0)}(x,t;\xi)$ is well-defined by \eqref{4.18}.

\begin{lemma}
For $(x,t)\in\mathcal{P}_\geq$, the function $M^{(0)}(x,t;\xi)$ is analytic for $\xi\in D_\epsilon(0)\setminus\Gamma^\epsilon$ such that $|M^{(0)}(x,t;\xi)|\leq C$.
Across the contour $\Gamma^\epsilon$, the continuous boundary values $M^{(0)}_\pm$ obey the jump relation $M_+^{(0)}=M_-^{(0)}J^{(0)}$, where jump matrix $J^{(0)}$ satisfies, for $1\leq n\leq\infty$,
\be\label{4.19}
\left\|J^{(1)}-J^{(0)}\right\|_{L^n(\Gamma^\epsilon)}\leq Ct^{-\frac{1}{3}(1+\frac{1}{n})}.
\ee
Moreover, as $t\rightarrow\infty$, we have
\be\label{4.20}
\left\|\left[M^{(0)}(x,t;\xi)\right]^{-1}-\mathbb{I}_{4\times4}\right\|_{L^\infty(\partial D_\epsilon(0))}=O\left(t^{-\frac{1}{3}}\right),
\ee
and
\be\label{4.21}
-\frac{1}{2\pi\ii}\int_{\partial D_\epsilon(0)}\left(\left[M^{(0)}(x,t;\xi)\right]^{-1}-\mathbb{I}_{4\times4}\right)\dd \xi=\frac{ M_1^{(0)}(y)}{(3t)^{\frac{1}{3}}}+O(t^{-\frac{2}{3}}),
\ee
where
\be\label{3.20}
\left(M_1^{(0)}(y)\right)_{12}=u_p(y),\quad \left(M_1^{(0)}(y)\right)_{14}=\ii w_p(y),
\ee
furthermore, the complex-valued function $u_p(y)$ and real-valued function $w_p(y)$ are the smooth solution of the new coupled Painlev\'e II equation \eqref{B.5}.
\end{lemma}
\begin{proof}
The proof follows the similar lines as Lemma 4.1 in \cite{LZG}.
\end{proof}
\subsection{Find asymptotic formula}
Define the contour $\hat{\Gamma}=\Gamma\cup\partial D_\epsilon(0)$ and let the boundary of $D_\epsilon(0)$ is oriented counterclockwise as depicted in Figure \ref{fig7}. We now introduce $\hat{M}(x,t;\xi)$ by
\be\label{4.23}
\hat{M}(x,t;\xi)=\left\{
\begin{aligned}
&M^{(1)}(x,t;\xi)\left[M^{(0)}(x,t;\xi)\right]^{-1},\quad \xi\in D_\epsilon(0),\\
&M^{(1)}(x,t;\xi),\qquad\qquad\qquad\,\,\,\qquad \xi\in\bfC\setminus D_\epsilon(0).
\end{aligned}
\right.
\ee
It then can be shown that $\hat{M}(x,t;\xi)$ satisfies the following RH problem:\\
$\bullet$ $\hat{M}(x,t;\xi)$ is analytic outside the contour $\hat{\Gamma}$ with continuous boundary values on $\hat{\Gamma}$;\\
$\bullet$ For $\xi\in\hat{\Gamma}$, we have the jump relation
$\hat{M}_+(x,t;\xi)=\hat{M}_-(x,t;\xi)\hat{J}(x,t;\xi);$\\
$\bullet$ $\hat{M}(x,t;\xi)\to\mathbb{I}_{4\times4}$, as $\xi\to\infty$;\\
where the jump matrix $\hat{J}(x,t;\xi)$ is described by
\be\label{4.24}
\hat{J}=\left\{
\begin{aligned}
&M^{(0)}_-J^{(1)}\left[M^{(0)}_+\right]^{-1},\,\,\ \xi\in\hat{\Gamma}\cap D_\epsilon(0),\\
&\left[M^{(0)}\right]^{-1},\qquad\qquad\ \xi\in\partial D_\epsilon(0),\\
&J^{(1)},\qquad\qquad\qquad\quad\ \xi\in\hat{\Gamma}\setminus\overline{D_\epsilon(0)}.
\end{aligned}
\right.
\ee
\begin{lemma}\label{lemma4.3}
Let $\hat{\omega}=\hat{J}-\mathbb{I}_{4\times4}$, $\tilde{\Gamma}=\Gamma\setminus(\bfR\cup\overline{D_\epsilon(0)})$. For each $1\leq n\leq\infty$, we have
\begin{align}
&\|\hat{\omega}\|_{L^n(\partial D_\epsilon(0))}\leq Ct^{-\frac{1}{3}},\label{4.25}\\
&\|\hat{\omega}\|_{L^n(\Gamma^\epsilon)}\leq Ct^{-\frac{1}{3}(1+\frac{1}{n})},\label{4.26}\\
&\|\hat{\omega}\|_{L^n(\bfR\setminus[-\xi_0,\xi_0])}\leq Ct^{-\frac{3}{2}},\label{4.27}\\
&\|\hat{\omega}\|_{L^n(\tilde{\Gamma})}\leq C\e^{-ct}.\label{4.28}
\end{align}
\end{lemma}
\begin{proof}
See the proof of Lemma 4.3 in \cite{CL}.
\end{proof}
As the discussion in Subsection \ref{sec3.4}, the estimates in Lemma \ref{lemma4.3} show that the RH problem for $\hat{M}$ has a unique solution given by
\be\label{4.29}
\hat{M}(x,t;\xi)=\mathbb{I}_{4\times4}+\frac{1}{2\pi\ii}\int_{\hat{\Gamma}}
(\hat{\nu}\hat{\omega})(x,t;s)\frac{\dd s}{s-\xi},
\ee
where $\hat{\nu}=\mathbb{I}_{4\times4}+(1-\hat{C}_{\hat{\omega}})^{-1}
\hat{C}_{\hat{\omega}}\mathbb{I}_{4\times4}$ satisfies the estimate
\be\label{4.30}
\|\hat{\nu}(x,t;\cdot)-\mathbb{I}_{4\times4}\|_{L^2(\hat{\Gamma})}=O(t^{-\frac{1}{3}}),\quad t\rightarrow\infty.
\ee
As $t\rightarrow\infty$, it then follows from \eqref{4.13}, \eqref{4.23} and \eqref{4.29} that
\be\label{4.31}
\begin{aligned}
\lim_{\xi\rightarrow\infty}\xi\left(M(x,t;\xi)-\mathbb{I}_{4\times4}\right)
&=\lim_{\xi\rightarrow\infty}\xi\left(\hat{M}(x,t;\xi)-\mathbb{I}_{4\times4}\right)\\
&=-\frac{1}{2\pi\ii}\int_{\hat{\Gamma}}(\hat{\nu}\hat{\omega})(x,t;\xi)\dd \xi\\
&=-\frac{1}{2\pi\ii}\int_{\partial D_\epsilon(0)}\hat{\omega}(x,t;\xi)\dd \xi+O(t^{-\frac{2}{3}})\\
&=\frac{ M_1^{(0)}(y)}{(3t)^{\frac{1}{3}}}+O(t^{-\frac{2}{3}}).
\end{aligned}
\ee
Using the reconstruction formula \eqref{2.22}, we hence obtain long-time asymptotics of the solution in Painlev\'e sector $\mathcal{P}$.
\begin{theorem}\label{th4.2}
Under the assumptions of Theorem \ref{th3.2}, the solution of the new two-component Sasa--Satsuma equation \eqref{SS} satisfies the following asymptotic formula in Painlev\'e region $\mathcal{P}$ as $t\to\infty$
\begin{align}
u(x,t)=&\frac{2\ii}{(3t)^{\frac{1}{3}}}u_p\left(\frac{x}{(3t)^{\frac{1}{3}}}\right)+O(t^{-\frac{2}{3}}),\\ w(x,t)=&-\frac{2}{(3t)^{\frac{1}{3}}}w_p\left(\frac{x}{(3t)^{\frac{1}{3}}}\right)+O(t^{-\frac{2}{3}}),
\end{align}
where complex-valued function $u_p(y)$ and real-valued function $w_p(y)$ denote the smooth solution of the new coupled Painlev\'e II equation \eqref{B.5} corresponding to $p\doteq\gamma(0)$ according to Lemma \ref{lemmaB.1}. Particularly, the function $u_p(y)$ has constant phase, namely, $\arg u_p$ is independent of $y$.
\end{theorem}
\section{Asymptotic analysis in fast decay sector $\mathcal F$}\label{sec5}

Finally, we will study the long-time asymptotic behavior of solution to Equation \eqref{SS} in the fast decay region $\mathcal F$ defined by
\be\label{5.1}
\mathcal F=\{(x,t)\in\bfR^2|t>1,\  Nt<-x,\ t\to\infty\},\ N \ \text{constant}.
\ee
In this region, the signature table for real part of phase function $\Phi(x,t;\xi)$ is shown in Figure \ref{fig8}.
\begin{figure}[htbp]
  \centering
  \includegraphics[width=3.5in]{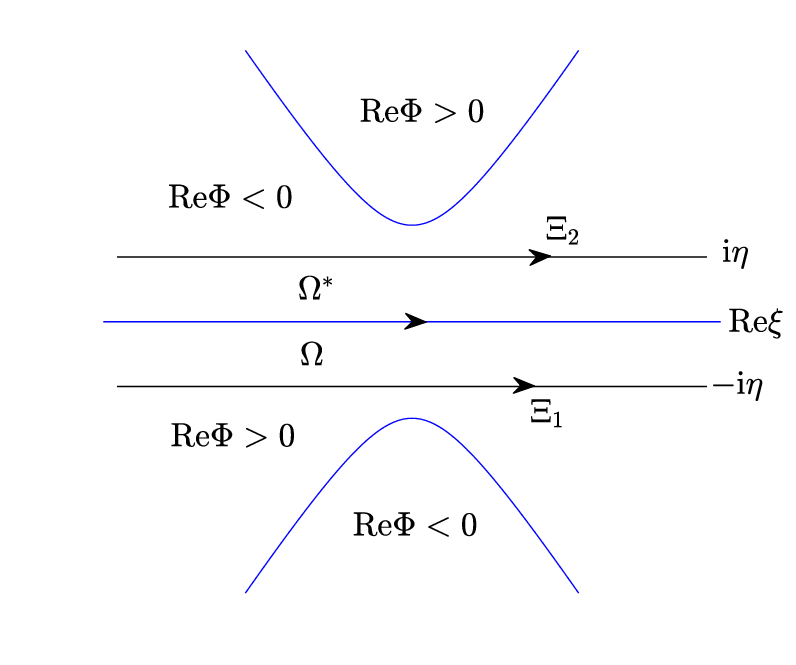}
  \caption{The contour $\Xi_1\cup\Xi_2$ and open sets $\Omega$, $\Omega^*$.}\label{fig8}
\end{figure}
\subsection{Transformations}
Define open sets
\be
\Omega=\{\xi\in\bfC|\text{Im}\xi\in(-\eta,0)\},\quad \Omega^*=\{\xi\in\bfC|\text{Im}\xi\in(0,\eta)\},
\ee
as shown in Figure \ref{fig8}. Then, according to the analysis of previous sections, we can similarly obtain the analytic decomposition of $\gamma(\xi)$: $\gamma(\xi)=\gamma_a(x,t;\xi)+\gamma_r(x,t;\xi)$. Moreover, the function $\gamma_{a}(x,t;\xi)$ is defined and continuous for $\xi\in\bar{\Omega}$ and analytic in $\Omega$, the $L^1, L^2$ and $L^\infty$ norms of the function $\gamma_{r}(x,t;\cdot)$ on $\bfR$ are $O(t^{-3/2})$ as $t\rightarrow\infty$. In order to perform the contour deformation, we would like to define the lines $\Xi_1=\{\xi\in\bfC|\text{Im}\xi=-\eta\}$, $\Xi_2=\{\xi\in\bfC|\text{Im}\xi=\eta\}$, and the orientation is depicted in Figure \ref{fig8}. Then, let us perform the transform
\be
M^{(1)}(x,t;\xi)=M(x,t;\xi)\times\left\{
\begin{aligned}
&\begin{pmatrix}
1 & \gamma_a(x,t;\xi)\e^{-t\Phi(x,t;\xi)}\\
\textbf{0}_{3\times1} & \mathbb{I}_{3\times3}
\end{pmatrix},\quad\ \ \xi\in \Omega,\\
&\begin{pmatrix}
1 & \textbf{0}_{1\times3}\\
-\gamma^\dag_a(x,t;\xi^*)\e^{t\Phi(x,t;\xi)} & \mathbb{I}_{3\times3}
\end{pmatrix},\quad \xi\in \Omega^*,\\
&\mathbb{I}_{4\times4},\qquad\qquad\qquad\qquad\quad\qquad\,\,\ \ \ \text{elsewhere}.
\end{aligned}
\right.
\ee
Hence, the matrix RH problem for $M^{(1)}(x,t;\xi)$ is as follows:\\
$\bullet$ $M^{(1)}(x,t;\xi)$ is analytic for $\xi\in\bfC\setminus(\bfR\cup\Xi_1\cup\Xi_2)$;\\
$\bullet$ The continuous boundary values $M_\pm(x,t;\xi)$ at $\bfR\cup\Xi_1\cup\Xi_2$ satisfy the jump condition
$M^{(1)}_+(x,t;\xi)=M^{(1)}_-(x,t;\xi)J^{(1)}(x,t;\xi)$;\\
$\bullet$ $M^{(1)}(x,t;\xi)=\mathbb{I}_{4\times4}+O(\xi^{-1})$, as $\xi\to\infty$;\\
where the jump matrix
\be
J^{(1)}(x,t;\xi)=\left\{
\begin{aligned}
&\begin{pmatrix}
1 & \gamma_a(x,t;\xi)\e^{-t\Phi(x,t;\xi)}\\
\textbf{0}_{3\times1} & \mathbb{I}_{3\times3}
\end{pmatrix},\quad \xi\in \Xi_1,\\
&\begin{pmatrix}
1 & \textbf{0}_{1\times3}\\
\gamma^\dag_a(x,t;\xi^*)\e^{t\Phi(x,t;\xi)} & \mathbb{I}_{3\times3}
\end{pmatrix},\quad\xi\in \Xi_2,\\
&\begin{pmatrix}
1 & \gamma_r(x,t;\xi)\e^{-t\Phi(x,t;\xi)}\\
\textbf{0}_{3\times1} & \mathbb{I}_{3\times3}
\end{pmatrix}
\begin{pmatrix}
1 & \textbf{0}_{1\times3}\\
\gamma^\dag_r(x,t;\xi^*)\e^{t\Phi(x,t;\xi)} & \mathbb{I}_{3\times3}
\end{pmatrix},\quad \xi\in\bfR.
\end{aligned}
\right.
\ee
\subsection{Find asymptotic formula }
Now $J^{(1)}(x,t;\xi)$ decays exponentially fast to the identity matrix $\mathbb{I}_{4\times4}$ as $t\to\infty$ on the contours $\Xi_1\cup\Xi_2$. Set $\omega=J^{(1)}-\mathbb{I}_{4\times4}$. Therefore, one can find that for $1\leq n\leq\infty$
\begin{align}
&\|\omega\|_{L^n(\bfR)}\leq Ct^{-\frac{3}{2}},\\
&\|\omega\|_{L^n(\Xi_1\cup\Xi_2)}\leq C\e^{-ct},
\end{align}
which immediately yields that
\be
\begin{aligned}
\lim_{\xi\rightarrow\infty}\xi\left(M(x,t;\xi)-\mathbb{I}_{4\times4}\right)
&=\lim_{\xi\rightarrow\infty}\xi\left(M^{(1)}(x,t;\xi)-\mathbb{I}_{4\times4}\right)\\
&=-\frac{1}{2\pi\ii}\int_{\bfR\cup\Xi_1\cup\Xi_2}(\mu\omega)(x,t;\xi)\dd \xi=O\left(t^{-\frac{3}{2}}\right),
\end{aligned}
\ee
where $\mu(x,t;\xi)$ is defined by $\mu=\mathbb{I}_{4\times4}+(1-\hat{C}_{\omega})^{-1}
\hat{C}_{\omega}\mathbb{I}_{4\times4}$ and obeys
\be
\|\mu(x,t;\cdot)-\mathbb{I}_{4\times4}\|_{L^2(\bfR\cup\Xi_1\cup\Xi_2)}=O\left(t^{-\frac{3}{2}}\right),\quad t\rightarrow\infty.
\ee
Hence, by \eqref{2.22}, we get the following theorem.
\begin{theorem}\label{th5.1}
In the fast decay region $\mathcal F$ described in \eqref{5.1} and under the conditions of Theorem \ref{th3.2}, the solution of the new two-component Sasa--Satsuma equation \eqref{SS} has the following long-time asymptotic behavior as $t\to\infty$
\be
u(x,t)=O\left(t^{-\frac{3}{2}}\right), \quad w(x,t)=O\left(t^{-\frac{3}{2}}\right).
\ee
\end{theorem}

\appendix
\section{ Proof of Theorem \ref{th3.1}}\label{secA}
The proof of Theorem \ref{th3.1} relies on deriving an explicit formula for the solution $M^X$ in terms of parabolic cylinder functions. First, note that the jump matrix $J^X$ obeys the symmetry
\be\label{A.1}
J^X(\rho;z)=\left(J^X\right)^\dag(\rho;z^*).
\ee
It then follows that the RH problem for $M^X(\rho;z)$ admits a Zhou's vanishing lemma \cite{ZX}, as a result, there exists a unique solution $M^X(\rho;z)$ which admits an expansion of the form \eqref{3.4} with respect to $z$.

To find the leading-order coefficient of large $z$ asymptotic behavior for the solution $M^X$, we let
\be
\Psi(z)=M^X(z)z^{\ii\nu\sigma}\e^{\frac{\ii z^2}{4}\sigma},
\ee
where we suppress the $\rho$ dependence for clarity. It follows that $\frac{\dd\Psi(z)}{\dd z}\Psi^{-1}(z)$ has no jump discontinuity along any of the rays. On the other hand, one find as $z\to\infty$,
\be
\frac{\dd\Psi(z)}{\dd z}\Psi^{-1}(z)=\frac{\ii }{2}\sigma z-\frac{\ii}{2}[\sigma,M_1^X]+O\left(z^{-1}\right).
\ee
Therefore, Liouville's argument implies that
\be\label{A.4}
\frac{\dd\Psi(z)}{\dd z}-\frac{\ii z}{2}\sigma\Psi(z)=\beta\Psi(z),
\ee
where
\be\label{A.5}
\beta=-\frac{\ii}{2}[\sigma,M_1^X]=\begin{pmatrix}
0~& \beta_{12}\\
\beta_{21}~& \mathbf{0}_{3\times3}\\
\end{pmatrix}.
\ee
Here we write a $4\times4$ matrix $A$ as a block form
\be
A=\begin{pmatrix}
A_{11} & A_{12}\\
A_{21} & A_{22}
\end{pmatrix}
\ee
with $A_{11}$ is scalar.
Particularly, we have
\be
(M_1^X)_{12}=-\ii\beta_{12},\quad (M_1^X)_{21}=\ii\beta_{21}.
\ee
The symmetries \eqref{A.1} of $J^X$ together with the uniqueness of the solution of the RH problem imply the following symmetry for $M^X$:
\be
\left[M^X(z)\right]^{-1}=\left[M^X(z^*)\right]^\dag,
\ee
which further yields that
\be
\beta_{21}=-\beta^\dag_{12}.
\ee

Next, considering \eqref{A.4} and \eqref{A.5}, we can obtain
\begin{align}
&\frac{\dd^2\Psi_{11}(z)}{\dd z^2}+\left(\frac{\ii}{2}+\frac{z^2}{4}-\beta_{12}\beta_{21}\right)\Psi_{11}(z)=0,\label{A.9}\\
&\beta_{12}\Psi_{21}(z)=\frac{\dd\Psi_{11}(z)}{\dd z}+\frac{\ii}{2}z\Psi_{11}(z),\label{A.10}\\
&\frac{\dd^2\beta_{12}\Psi_{22}(z)}{\dd z^2}-\left(\frac{\ii}{2}-\frac{z^2}{4}+\beta_{12}\beta_{21}\right)\beta_{12}\Psi_{22}(z)=0,\label{A.11}\\
&\Psi_{12}(z)=\frac{1}{\beta_{12}\beta_{21}}\left(\frac{\dd\beta_{12}\Psi_{22}(z)}{\dd z}-\frac{\ii}{2}z\beta_{12}\Psi_{22}\right).\label{A.12}
\end{align}
Then, by simple change of variables, it can be shown that Equations \eqref{A.9} and \eqref{A.11} can be transformed into the parabolic cylinder equation
\be\label{A.13}
\frac{\dd^2g(k)}{\dd k^2}+\left(\frac{1}{2}-\frac{k^2}{4}+a\right)g(k)=0.
\ee
However, it is known that the solution of Equation \eqref{A.13} can be expressed as
\be
g(k)=\nu_1D_a(k)+\nu_2D_a(-k),
\ee
where $\nu_1$, $\nu_2$ are constants and $D_a(\cdot)$ is the standard parabolic cylinder function \cite{WW}. Then, denoting $a=\ii\beta_{12}\beta_{21}$, we have
\begin{align}
\Psi_{11}(z)=&\nu_1D_a(\e^{-\frac{3\ii\pi}{4}}z)+\nu_2D_a(\e^{\frac{\pi\ii}{4}}z),\\
\beta_{12}\Psi_{22}(z)=&\nu_3D_{-a}(\e^{\frac{3\ii\pi}{4}}z)+\nu_4D_{-a}(\e^{-\frac{\pi\ii}{4}}z).
\end{align}
On the other hand, it follows from \cite{WW} that as $k\rightarrow\infty$,
\be\label{A.17}
D_a(k)=\left\{\begin{aligned}
&k^a\e^{-\frac{k^2}{4}}(1+O(k^{-2})),\qquad\qquad\qquad\qquad\quad\qquad\qquad\qquad\qquad|\arg k|<\frac{3\pi}{4},\\
&k^a\e^{-\frac{k^2}{4}}(1+O(k^{-2}))-\frac{\sqrt{2\pi}}{\Gamma(-a)}\e^{a\pi\ii+\frac{k^2}{4}}k^{-a-1}
(1+O(k^{-2})),\qquad\frac{\pi}{4}<\arg k<\frac{5\pi}{4},\\
&k^a\e^{-\frac{k^2}{4}}(1+O(k^{-2}))-\frac{\sqrt{2\pi}}{\Gamma(-a)}\e^{-a\pi\ii+\frac{k^2}{4}}k^{-a-1}(1+O(k^{-2})),
\ -\frac{5\pi}{4}<\arg k<-\frac{\pi}{4},
\end{aligned}
\right.
\ee
Hence, as $\arg z\in(-\frac{3\pi}{4},-\frac{\pi}{4})$, we find that
\begin{align}
\Psi_{11}(z)&=\e^{-\frac{\pi\nu}{4}}D_a(\e^{\frac{\pi\ii}{4}}z),~a=-\ii\nu,\label{A.18}\\
\beta_{12}\Psi_{22}(z)&=\beta_{12}\e^{\frac{3\pi\nu}{4}}D_{-a}(\e^{\frac{3\pi\ii}{4}}z),\label{A.19}
\end{align}
because as $z\rightarrow\infty$,
\be
\Psi_{11}\rightarrow z^{-\ii\nu}\e^{-\frac{\ii z^2}{4}},\quad \Psi_{22}\rightarrow z^{\ii\nu}\e^{\frac{\ii z^2}{4}}\mathbb{I}_{3\times3}.
\ee
It then follows from the property of $D_a(\cdot)$
\begin{align}
\frac{\dd D_a(k)}{\dd k}+\frac{k}{2}D_a(k)-aD_{a-1}(k)=0,
\end{align}
and \eqref{A.10}, \eqref{A.12} that
\be
\begin{aligned}
\beta_{12}\Psi_{21}(z)&=a\e^{\frac{\pi(\ii-\nu)}{4}}D_{a-1}(\e^{\frac{\pi\ii}{4}}z),\\
\Psi_{12}(z)&=\beta_{12}\e^{\frac{\pi(\ii+3\nu)}{4}}D_{-a-1}(\e^{\frac{3\pi\ii}{4}}z).
\end{aligned}
\ee
Accordingly, for $\arg z\in(-\frac{\pi}{4},\frac{\pi}{4})$, we can get
\begin{align}
\Psi_{11}(z)&=\e^{-\frac{\pi\nu}{4}}D_a(\e^{\frac{\pi\ii}{4}}z),~a=-\ii\nu,\\
\beta_{12}\Psi_{22}(z)&=\beta_{12}\e^{-\frac{\pi\nu}{4}}D_{-a}(\e^{-\frac{\pi\ii}{4}}z),\\
\beta_{12}\Psi_{21}(z)&=a\e^{\frac{\pi(\ii-\nu)}{4}}D_{a-1}(\e^{\frac{\pi\ii}{4}}z),\\
\Psi_{12}(z)&=\beta_{12}\e^{-\frac{\pi(3\ii+\nu)}{4}}D_{-a-1}(\e^{-\frac{\pi\ii}{4}}z).
\end{align}

Now, since across the ray $\arg z=-\frac{\pi}{4}$,
\be
\begin{aligned}
\Psi_+(z)=&M^X_+(z)z^{\ii\nu\sigma}\e^{\frac{\ii z^2}{4}\sigma}=M^X_-(z)J^X(z)z^{\ii\nu\sigma}\e^{\frac{\ii z^2}{4}\sigma}\\
=&\Psi_-(z)z^{-\ii\nu\sigma}\e^{-\frac{\ii z^2}{4}\sigma}J^X(z)z^{\ii\nu\sigma}\e^{\frac{\ii z^2}{4}\sigma}
=\Psi_-(z)\begin{pmatrix}
1 & \rho\\
\mathbf{0}_{3\times1} & \mathbb{I}_{3\times3}
\end{pmatrix},
\end{aligned}
\ee
thus,
\be
\beta_{12}\e^{-\frac{\pi(3\ii+\nu)}{4}}D_{-a-1}(\e^{-\frac{\pi\ii}{4}}z)
=\e^{-\frac{\pi\nu}{4}}D_a(\e^{\frac{\pi\ii}{4}}z)\rho
+\beta_{12}\e^{\frac{\pi(\ii+3\nu)}{4}}D_{-a-1}(\e^{\frac{3\pi\ii}{4}}z).
\ee
However, note from \cite{WW} that
\be
D_{a}(\e^{\frac{\pi\ii}{4}}z)=\frac{\Gamma(a+1)}{\sqrt{2\pi}}
\left(\e^{\frac{\ii\pi a}{2}}D_{-a-1}(\e^{\frac{3\pi\ii}{4}}z)+\e^{-\frac{\ii\pi a}{2}}D_{-a-1}(\e^{-\frac{\pi\ii}{4}}z)\right),
\ee
therefore, we can find that
\be
\beta_{12}=\frac{\Gamma(-\ii\nu)}{\sqrt{2\pi}}\e^{\frac{\pi\ii}{4}-\frac{\pi\nu}{2}}\nu\rho.
\ee
The estimate \eqref{3.6} is an consequence of the explicit solution $M^X(\rho;z)$ which is expressed in terms of $D_{\tilde{a}}(z)$.

\section{A new coupled Painlev\'e II RH problem}\label{secB}
\begin{figure}[htbp]
\centering
\includegraphics[width=3.5in]{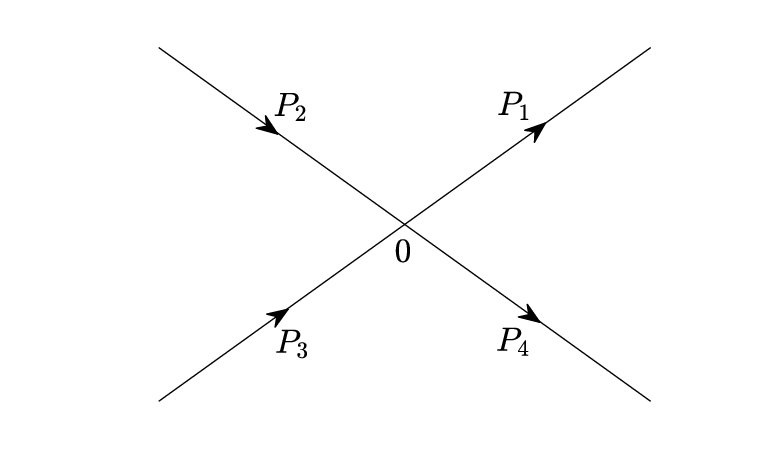}
\caption{The oriented contour $P=\cup_{j=1}^4P_j$.}\label{fig9}
\end{figure}
Let $P$ denote the contour $P=\cup_{j=1}^4P_j$ oriented to the right as in Figure \ref{fig9}, where
\be
\begin{aligned}
P_1=&\{l\e^{\frac{\pi\ii}{6}}|0\leq l<\infty\},\qquad P_2=\{l\e^{\frac{5\pi\ii}{6}}|0\leq l<\infty\},\\
P_3=&\{l\e^{-\frac{5\pi\ii}{6}}|0\leq l<\infty\},\quad P_4=\{l\e^{-\frac{\pi\ii}{6}}|0\leq l<\infty\}.
\end{aligned}
\ee
Let $p$ be a $1\times3$ complex-valued row vector with $p=p^*\sigma_1$ and define the jump matrix by
\be\label{B.2}
J^P(y;z)=\left\{
\begin{aligned}
&\begin{pmatrix}
1 & \textbf{0}_{1\times3}\\
p^\dag\e^{2\ii(\frac{4}{3}z^3-yz)} & \mathbb{I}_{3\times3}
\end{pmatrix},\quad z\in P_1\cup P_2,\\
&\begin{pmatrix}
1 & p\e^{-2\ii(\frac{4}{3}z^3-yz)}\\
\textbf{0}_{3\times1} & \mathbb{I}_{3\times3}
\end{pmatrix},\quad z\in P_3\cup P_4.
\end{aligned}
\right.
\ee
Then we consider the following model RH problem:\\
$\bullet$ $M^P(y;z)$ is analytic in $\bfC\setminus P$ with continuous boundary values on $P$;\\
$\bullet$ $M^P_+(y;z)=M^P_-(y;z)J^P(y;z)$, for $z\in P$;\\
$\bullet$ $M^P(y;z)\rightarrow \mathbb{I}_{4\times4}$, as $z\rightarrow\infty$.
\begin{lemma}\label{lemmaB.1}
The RH problem for $M^P(y;z)$ has a unique solution for each $y\in\bfR$. Moreover, there are smooth functions $\{M_j^P(y)\}$ of $y\in\bfR$ with decay as $y\rightarrow-\infty$ such that
\be\label{B.3}
M^P(y;z)=\mathbb{I}_{4\times4}+\sum_{j=1}^N\frac{M_j^P(y)}{z^j}+O(z^{-N-1}), \quad z\rightarrow\infty,
\ee
and the $(1,2)$, $(1,4)$ entries of leading coefficient $M_1^P$ can be expressed by
\be\label{B.4}
\left(M_1^P(y)\right)_{12}=u_p(y),\quad \left(M_1^P(y)\right)_{14}=\ii w_p(y),
\ee
where $u_p(y)$ and $w_p(y)$ are complex-valued and real-valued functions, respectively, and satisfy a new coupled Painlev\'e II equation
\be\label{B.5}
\begin{aligned}
u_p^{''}(y)+8\left(2|u_p(y)|^2+w_p^{2}(y)\right)u_p(y)+yu_p(y)=0,\\
w_p^{''}(y)+8\left(2|u_p(y)|^2+w_p^{2}(y)\right)w_p(y)+yw_p(y)=0.
\end{aligned}
\ee
Moreover, the function $u_p(y)$ has constant phase, that is, $\arg u_p$ is independent of $y$.
\end{lemma}
\begin{proof}
The symmetry
\be\label{B.6}
J^P(y;z)=\left[J^P(y;z^*)\right]^\dag,
\ee
implies that the jump condition and jump matrix in RH problem for $M^P(y;z)$ satisfy the hypotheses of Zhou's vanishing lemma \cite{ZX}. Therefore, the existence and uniqueness of $M^P(y;z)$ immediately follow, and hence the expansion \eqref{B.3}. On the other hand, by the standard nonlinear steepest descent analysis, the coefficients $M_j^P(y)$ have exponential decay as $y\rightarrow-\infty$.

Let $\psi(y;z)=M^P(y;z)\e^{\ii(\frac{4}{3}z^3-yz)\sigma}$. Then, the function $\mathcal{B}(y;z)=\psi_y(y;z)\psi^{-1}(y;z)$ is an entire function of $z$ by Liouville's theorem. Thus, one can find
\be\label{B.7}
\mathcal{B}=(M_y^P-\ii zM^P\sigma)\left(M^P\right)^{-1}=-\ii\sigma z+\ii[\sigma,M_1^P]\doteq\mathcal{B}_1z+\mathcal{B}_0.
\ee
Similarly, the function $\mathcal{C}(y;z)=\psi_z(y;z)\psi^{-1}(y;z)$ is entire, and hence,
\be\label{B.8}
\begin{aligned}
\mathcal{C}=&\left(M_z^P+\ii(4z^2-y)M^P\sigma\right)\left(M^P\right)^{-1}\\
=&4\ii\sigma z^2-4\ii[\sigma,M_1^P]z-4\ii[\sigma,M_{2}^P]+4\ii[\sigma,M_1^P]M_1^P-\ii y\sigma\\
\doteq&\mathcal{C}_2z^2+\mathcal{C}_1z+\mathcal{C}_0.
\end{aligned}
\ee
On the other hand, by \eqref{B.7}, we also have
\be\label{B.9}
M_y^P-\ii zM^P\sigma=\mathcal{B}M^P,
\ee
and then inserting the expansion \eqref{B.3} into \eqref{B.9}, one can obtain
\be\label{B.10}
\ii[\sigma,M_2^P]=\ii[\sigma,M_1^P]M_1^P-M_{1y}^P,
\ee
that is,
\be
\mathcal{C}_0=4M_{1y}^P-\ii y\sigma.
\ee
The definitions of $\mathcal{B}$ and $\mathcal{C}$ yields that the function $\psi$ admits the Lax equations
\be\label{B.12}
\left\{
\begin{aligned}
\psi_y=&\mathcal{B}\psi,\\
\psi_z=&\mathcal{C}\psi.
\end{aligned}
\right.
\ee
Then the compatibility condition of \eqref{B.12} implies that
\be\label{B.13}
4M^P_{1yy}-4\ii[\sigma,M_1^P]M_{1y}^P+4\ii M_{1y}^P[\sigma,M_1^P]-y[\sigma,M_1^P]\sigma+y\sigma[\sigma,M_1^P]=0.
\ee
However, the symmetric relation \eqref{B.6} and
\be\label{B.14}
J^P(y;z)=\mathcal{A}\left[J^P(y;-z^*)\right]^*\mathcal{A}
\ee
implies that the solution of RH problem for $M^P(y;z)$ satisfies
\be\label{B.15}
M^P(y;z)=\left[\left(M^P\right)^\dag(y;z^*)\right]^{-1}=\mathcal{A}\left[M^P(y;-z^*)\right]^*\mathcal{A}.
\ee
Hence, the leading order coefficient $M_1^P(y)$ obeys
\be\label{B.16}
-\left[M_1^P(y)\right]^\dag=M_1^P(y)=-\mathcal{A}\left[M_1^P(y)\right]^*\mathcal{A}.
\ee
For convenience, we write
\be\label{B.17}
M_1^P(y)=\begin{pmatrix}
\left(M_1^P(y)\right)_{11} & \left(M_1^P(y)\right)_{12}\\[4pt]
\left(M_1^P(y)\right)_{21} & \left(M_1^P(y)\right)_{22}
\end{pmatrix},
\ee
where $\left(M_1^P(y)\right)_{11}$ is scalar.
Now, substituting \eqref{B.17} into \eqref{B.13}, using the first relation in \eqref{B.16}, we have
\begin{align}
&\left(M_1^P(y)\right)''_{11}-2\ii\left(M_1^P(y)\right)_{12}\left[\left(M_1^P(y)\right)'_{12}\right]^\dag
-2\ii\left(M_1^P(y)\right)'_{12}\left(M_1^P(y)\right)^\dag_{12}=0,\label{B.18}\\
&\left(M_1^P(y)\right)''_{12}+2\ii\left(M_1^P(y)\right)_{12}\left(M_1^P(y)\right)'_{22}
-2\ii\left(M_1^P(y)\right)'_{11}\left(M_1^P(y)\right)_{12}+y\left(M_1^P(y)\right)_{12}=0,\label{B.19}\\
&\left(M_1^P(y)\right)''_{22}+2\ii\left(M_1^P(y)\right)_{12}^\dag\left(M_1^P(y)\right)'_{12}
+2\ii\left[\left(M_1^P(y)\right)'_{12}\right]^\dag\left(M_1^P(y)\right)_{12}=0.\label{B.20}
\end{align}
Since $M^P_1(y)$ and its derivatives decay as $y\to-\infty$, it follows from \eqref{B.18} and \eqref{B.20} that
\be\label{B.21}
\begin{aligned}
\left(M_1^P(y)\right)'_{11}=&2\ii\left[\left(M_1^P(y)\right)_{12}\left(M_1^P(y)\right)^\dag_{12}\right],\\ \left(M_1^P(y)\right)'_{22}=&-2\ii\left[\left(M_1^P(y)\right)^\dag_{12}\left(M_1^P(y)\right)_{12}\right].
\end{aligned}
\ee
Inserting \eqref{B.21} into \eqref{B.19}, we get
\be\label{B.22}
\left(M_1^P(y)\right)''_{12}+8\left(M_1^P(y)\right)_{12}\left[\left(M_1^P(y)\right)^\dag_{12}\left(M_1^P(y)\right)_{12}\right]
+y\left(M_1^P(y)\right)_{12}=0.
\ee
Finally, according to the second symmetry in \eqref{B.16}, we find $\left(M_1^P(y)\right)_{12}=-\left[\left(M_1^P(y)\right)_{12}\right]^*\sigma_1$, thus, we can write
\be\label{B.23}
\left(M_1^P(y)\right)_{12}=\begin{pmatrix}
u_p(y) & -u^*_p(y)  & \ii w_p(y)
\end{pmatrix},
\ee
where $u_p(y)$ and $w_p(y)$ are complex-valued and real-valued functions, respectively. Substituting \eqref{B.23} into \eqref{B.22}, one immediately obtain
\be\label{B.24}
\begin{aligned}
u_p^{''}(y)+8\left(2|u_p(y)|^2+w_p^{2}(y)\right)u_p(y)+yu_p(y)=0,\\
w_p^{''}(y)+8\left(2|u_p(y)|^2+w_p^{2}(y)\right)w_p(y)+yw_p(y)=0.
\end{aligned}
\ee
Writing $u_p(y)=f(y)\e^{\ii\alpha(y)}$ with $f(y)$, $\alpha(y)$ being real functions, then the first equation in \eqref{B.24} reduces to the following system
\begin{align}
f''+8(2f^2+w_p^2)f+yf-f(\alpha')^2=&0,\label{B.25}\\
2f'\alpha'+f\alpha''=&0.\label{B.26}
\end{align}
It then follows from \eqref{B.26}  that
\be
f^2\alpha'=C,
\ee
where $C$ is a real constant. Using this relation to eliminate $\alpha'$ from \eqref{B.25}, we find
\begin{align}
f''+8(2f^2+w_p^2)f+yf-C^2f^{-3}=0.
\end{align}
The decay of $u_p$, $w_p$ and their derivatives as $y\to-\infty$ imply that $C=0$. Therefore, $\alpha(y)=\arg u_p(y)$ is independent of $y$.

The proof of lemma is now completed.
\end{proof}

\medskip
\small{

}
\end{document}